\documentclass[lettersize,journal]{IEEEtran}
\usepackage{amsmath}
\usepackage{amsfonts}
\usepackage{algorithmic}
\usepackage{array}
\usepackage{amssymb}
\usepackage{textcomp}
\usepackage{stfloats}
\usepackage{url}
\usepackage{verbatim}
\usepackage{graphicx}
\usepackage{cite}

\usepackage{upgreek}
\usepackage{pifont}

\usepackage{soul}
\usepackage{multirow}
\usepackage[table]{xcolor}
\PassOptionsToPackage{hyphens}{url}
\usepackage{hyperref}
\definecolor{darkred}{rgb}{0.8, 0.0, 0.0}
\definecolor{darkgreen}{rgb}{0.0, 0.5, 0.0} 

\usepackage[ruled]{algorithm2e}
\usepackage{dsfont}
\usepackage{xcolor}
\usepackage{subcaption}
\allowdisplaybreaks[4]

\usepackage{makecell}
\usepackage{enumitem}
\usepackage{amsthm}
\newtheorem{theorem}{Theorem}

\newtheorem{remark}{Remark}
\newtheorem{proposition}{Proposition}
\newtheorem{assumption}{Assumption}
\newtheorem{lemma}{Lemma}

\ifodd 0
\newcommand{\rev}[1]{{\color{blue}#1}} 
\newcommand{\newrev}[1]{{\color{red}#1}} 
\else
\newcommand{\rev}[1]{#1}
\newcommand{\newrev}[1]{#1} 
\fi

\begin{document}

\title{Update the Unseen Only: Minimizing AoI for Collaborative Perception through Online Learning}

\author{Yanan Ma, Zhuoyi Zhao, Zhengru Fang, Haonan An, Xianhao Chen,~\IEEEmembership{Member,~IEEE}, \\
and Yuguang Fang,~\IEEEmembership{Fellow,~IEEE}
\thanks{
The work was supported in part by the JC STEM Lab of Smart City funded by The Hong Kong Jockey Club Charities Trust under Contract 2023-0108, in part by the Research Grants Council of the Hong Kong SAR, China (Project No. CityU 11216324), and in part by the Hong Kong SAR Government under the Global STEM Professorship.
The work of X. Chen was supported in part by the Research Grants Council of Hong Kong under Grant 27213824 and CRS HKU702/24.
}
\thanks{Y. Ma, Z. Fang, H. An, and Y. Fang are with Hong Kong JC STEM Lab of Smart City and the Department of Computer Science, City University of Hong Kong, Hong Kong, China. (e-mail: yananma8-c@my.cityu.edu.hk, zhefang4-c@my.cityu.edu.hk, haonanan2-c@my.cityu.edu.hk, my.fang@cityu.edu.hk.)}

\thanks{Z. Zhao and X. Chen are with the Department of Electrical and Computer Engineering, The University of Hong Kong, Hong Kong, China. (e-mail: zhuoyijoeyzhao@gmail.com, xchen@eee.hku.hk.)}%
}



\maketitle

\begin{abstract}
While collaborative perception (CP) enhances the safety of autonomous driving, limited bandwidth can cause severe shared data staleness in CP systems. Existing age-of-information (AoI) minimization policies are not well suited for CP, as they overlook the fact that a vehicle’s AoI decreases not only through updates from the source (i.e., a base station) but also through the vehicle’s \textit{local sensing}. To address this issue, we propose a mobility-aware AoI minimization framework for CP that explicitly accounts for vehicles' dynamic sensing ranges. We first derive a closed-form expression for the long-term time average sum AoI within a considered region, accommodating an ever-changing vehicle population and their dynamic sensed areas. 
Based on this characterization, we develop Local-sensing-aware Max-Weight Scheduling (LocMW), an online learning algorithm designed for sensor information broadcast from a source to vehicles under unknown environmental statistics and delayed observations. We provide performance guarantees demonstrating that LocMW achieves a sublinear cumulative excess AoI compared to the optimal stationary randomized benchmark. 
Extensive simulations using vehicular trajectory datasets and 3D perception tasks demonstrate that our LocMW policy substantially outperforms competing baselines, reducing the time-averaged sum AoI by up to 31.6\% and improving mAP detection accuracy by up to 16.3\%.
\end{abstract}


\section{Introduction}
\IEEEPARstart{T}{he} realization of Level-5 autonomous driving requires vehicles to perceive dynamic traffic environments with high precision. However, a single autonomous vehicle is vulnerable to various occlusions in dense urban scenarios, as it is fundamentally limited by the line-of-sight nature of onboard geometric sensors, such as cameras and LiDAR. To overcome this limitation,  collaborative perception (CP) has emerged as a promising solution \cite{fang2024pacp, ma2026birdcast, hu2025collaborative}. For instance, in an infrastructure-assisted CP system, a base station (BS) or roadside unit equipped with camera/LiDAR sensors can broadcast sensor information to surrounding vehicles to complement their local views~\cite{he2021vi}. By integrating observations from vehicles and road infrastructure, CP eliminates the blind spots of individual vehicles, thus improving the safety of autonomous driving~\cite{hu2025collaborative, chen2024vehicle, fang2024pacp, ma2026birdcast, ma2026sense4fl}.  

Despite its perception benefits, CP is constrained by the limited bandwidth available for sensor data transmissions. Prior works have addressed the communication bottleneck in CP by optimizing resource allocation based on traditional quality-of-service (QoS) metrics, such as network throughput \cite{fang2024pacp}, spectral efficiency, and communication volume \cite{hu2022where2comm}. Nevertheless, these metrics focus on transmission efficiency rather than on the staleness of information from the vehicles' perspectives, though the latter is a more important objective for real-time sensing systems.

\begin{figure}
    \centering
    \includegraphics[width=1.0\linewidth]{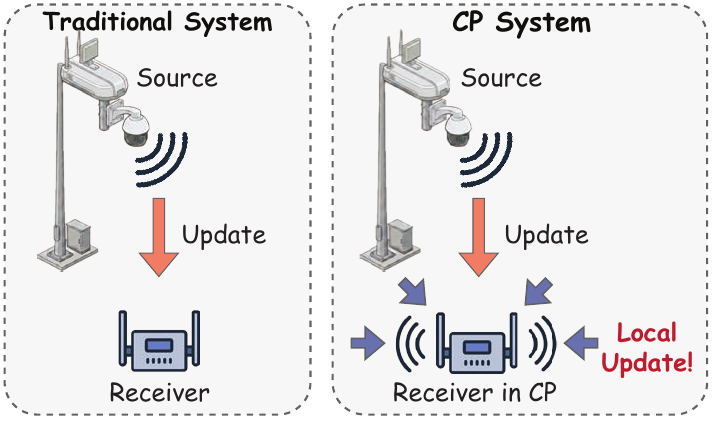}
    \caption{Comparison of information update mechanisms in traditional systems versus CP systems. In a traditional system, AoI decreases for the receiver that relies exclusively on direct transmissions from the source. In a CP system, the receiver is equipped with onboard sensors that allow it to actively perceive its surroundings, which also leads to AoI decrease.}
    \label{fig:comparison}
    \vspace{-0.0cm}
\end{figure}

To maintain information freshness, the age of information (AoI) metric has been widely adopted in sensor networks, which measures information freshness from the destination's perspective~\cite{kadota2018scheduling, yates2021age, kadota2019scheduling, kadota2019minimizing, fang2025r, han2026spatiotemporal, zhu2026timeliness, zhao2025optimizing1}. Optimizing AoI for CP ensures the delivery of fresh sensor data in autonomous driving, thereby enhancing driving safety. Yet, previous research efforts on AoI are ill-suited to CP systems. Specifically, prior work implicitly \textit{assumes that receivers lack sensing capabilities}, because their AoI decreases only upon receiving updates from a transmitter. In contrast, in CP systems, vehicles (or drones or robots) actively perceive their local environments (see Fig.~\ref{fig:comparison}). For instance, if a vehicle cannot see an occluded area, a traditional AoI-minimization scheduling policy may transmit information about this area. However, due to vehicle mobility and environmental changes, the vehicle may see this area directly, implying that data transmissions, if scheduled, are redundant. Conventional AoI scheduling policies that ignore receivers' local sensing capabilities can therefore lead to significantly suboptimal performance in CP systems.

The mobility and sensing dynamics in CP introduce unique challenges for AoI-minimization scheduling policies.  First, BS update scheduling depends not only on the current AoI, but also on the set of currently present vehicles and their local sensing coverage. Jointly characterizing and optimizing these time-varying states is challenging. This is because vehicles may enter or leave a target region over time, and their local sensing ranges change dynamically due to mobility, viewpoints, and occlusions. Second, the BS must make scheduling decisions with delayed knowledge of each vehicle's sensing conditions, because occlusion information is available only after perception and uplink feedback, both of which incur non-negligible delays. These challenges lead to two key research questions:
\begin{itemize}
\item \textbf{Q1.} \textit{How can we characterize a tractable AoI objective that captures both the time-varying receiver population and their dynamic local sensing ranges?}
\item \textbf{Q2.} \textit{How can we minimize the formulated AoI under unknown environmental parameters and delayed network status?}
\end{itemize}


To answer these questions, we develop a unified local-sensing-aware online scheduling framework for infrastructure-assisted CP. On the modeling side, we characterize the population of vehicles that are interested in a specific area but unable to directly observe it as a temporally correlated stochastic process. By exploiting the dynamics of this unobserving population, we derive a closed-form expression for the time-average sum AoI, accounting for the time-varying receiver population and dynamic local sensing ranges. 
On the optimization side, delayed state observations and unknown environmental parameters render the scheduling problem a partially observable restless multi-armed bandit (PO-RMAB) problem, which is generally PSPACE-hard. To address this challenge, we propose the Local-sensing-aware Max-Weight scheduling (LocMW) policy for real-time update scheduling, a learning-aided online policy that combines projected ridge estimation of environmental parameters, certainty-equivalent prediction over the observation-delay window, and Lyapunov-drift-based Max-Weight scheduling.
We show that when the system parameters are known, LocMW achieves a time-average sum AoI no older than the optimal stationary randomized benchmark and within a factor of two of the mean-field lower bound. 
With online learning, LocMW incurs a sublinear cumulative excess AoI relative to the optimal stationary randomized benchmark.
This is the first work to address the AoI minimization problem by explicitly accounting for receivers’ dynamic local sensing capabilities. While we specifically focus on vehicular CP systems in this paper, the proposed methodology can be broadly extended to AoI minimization problems in other types of freshness-critical CP systems, such as drone swarms and robotic networks.

The main contributions of this paper are summarized as follows:
\begin{itemize}
    \item We propose the first AoI minimization framework that explicitly accounts for receivers’ local sensing capabilities and the time-varying unobserving populations in CP systems. Based on this framework, we derive a closed-form characterization of the time-average sum AoI over a target region, revealing how mobility-driven local sensing reshapes the information-staleness dynamics.

    \item Through a mean-field relaxation, we derive an analytical lower bound on the time-average sum AoI. We further characterize the optimal stationary randomized benchmark and show that it is within a constant factor of this lower bound.
    
    \item We design LocMW scheduling, an online learning policy for sensor information broadcasting to navigate observation delays and unknown environmental dynamics. We also show that it incurs a sublinear cumulative excess AoI 
    against the optimal stationary randomized benchmark.

    \item We conduct extensive simulations using real-world vehicular trajectory datasets (pNEUMA and FLUID) and a V2X CP dataset (V2X-Sim) to evaluate our framework. The results demonstrate that LocMW consistently outperforms other baselines by reducing information staleness and improving perception performance, achieving up to a 37.5\% reduction in AoI and a 6.79\% gain in mean average precision (mAP).
\end{itemize}

The remainder of this paper is organized as follows. Section \ref{sec:related_work} reviews related work on AoI scheduling and resource management in CP. Section \ref{sec:system_model} presents the system model and problem formulation. Section \ref{sec:aoi_def} characterizes the time-average sum AoI in closed form. Sections \ref{sec:lower_bound} and \ref{sec:randomized_policy} detail the analytical lower bound and the optimal stationary randomized benchmark, respectively. Section \ref{sec:online_learning} develops the LocMW algorithm and establishes its performance guarantee. Finally, Section \ref{sec:simulation} reports the simulation results, and Section \ref{sec:conclusion} concludes the paper.

\section{Related Work}
\label{sec:related_work}

\subsection{Age of Information}
AoI has become a standard metric for quantifying information freshness from the destination's perspective~\cite{kadota2018scheduling, sun2017update, xu2022aoi, abd2019role, zhao2025optimizing1}. A substantial body of research has investigated AoI minimization under various communication, scheduling, and system constraints. Kadota \textit{et al.}~\cite{kadota2018scheduling} formulated a discrete-time scheduling problem over unreliable channels and developed randomized, Max-Weight, and Whittle-index policies for minimizing the expected weighted-sum AoI. Sun \textit{et al.}~\cite{sun2017update} demonstrated that the zero-wait policy is not necessarily age-optimal and formulated average age-penalty minimization as a constrained semi-Markov decision process, based on which optimal causal update policies were derived. Ji \textit{et al.}~\cite{ji2024age} considered age-optimal packet scheduling under long-term resource constraints and delayed feedback, and proposed a low-complexity greedy policy that minimizes the immediate expected Lagrangian cost. Under constrained transmission rates and imperfect feedback, Zhu \textit{et al.}~\cite{zhu2026age} developed a Lyapunov-optimization-based drift-plus-penalty policy for systems with Bernoulli traffic. Tsai \textit{et al.}~\cite{tsai2023distribution} further studied update-through-queue systems with random, unknown delays and subsequently designed online algorithms to adaptively learn the optimal waiting time.
In addition to traditional AoI formulations, new AoI metrics have also been introduced to capture the usefulness of status updates. This includes age of changed information (AoCI), which accounts for changes in the underlying source state \cite{wang2021age}, and age of incorrect information (AoII) \cite{maatouk2020age}, which incorporates the discrepancy between the receiver's estimate and the true system state. However, these studies fail to account for users' local sensing capabilities and dynamic mobility, in which a \textit{time-varying} set of vehicles \textit{refresh} their information through onboard sensing.

Beyond traditional scheduling, recent studies have applied advanced learning and optimization techniques to manage AoI in highly dynamic environments. Chen \textit{et al.} \cite{chen2020age} optimized long-term AoI performance in a Manhattan grid vehicle-to-vehicle (V2V) network and proposed a decentralized deep reinforcement learning algorithm. Emami \textit{et al.} \cite{emami2024age} studied the trade-off between UAV mobility and data freshness as a mean-field game and developed a hybrid proximal policy optimization scheme to jointly optimize trajectories and communication schedules. 
Nevertheless, these schemes still do not consider users' local sensing capabilities.

In a nutshell, existing AoI-minimization frameworks cannot be directly applied to CP systems for two reasons: 1) they assume that AoI can only decrease via network transmissions without considering local sensing, and 2) they assume that the set of receivers is fixed, which does not hold in dynamic mobile networks such as vehicular networks.

\subsection{Collaborative Perception}
The transmission of high-dimensional sensory data over bandwidth-limited V2X links remains a primary bottleneck for CP. To alleviate this, a growing body of literature has investigated dynamic resource allocation, user selection, and task-oriented scheduling for communication-efficient CP \cite{ma2026birdcast, fang2024pacp, hu2022where2comm, wu2025fresh2comm, wang2023age, tao2025directed}.
To enhance communication efficiency, recent studies focus on selective transmission strategies that filter redundant data. For example, Where2comm \cite{hu2022where2comm} and How2comm \cite{yang2023how2comm} reduce communication overhead by transmitting high-uncertainty spatial regions. Other frameworks, such as PACP \cite{fang2024pacp} and Directed-CP \cite{tao2025directed}, further improve CP performance by dynamically prioritizing data based on perception correlation or directional interests under varying wireless conditions. To address the downlink communication bottleneck in vehicle-to-infrastructure (V2I) systems, Ma \textit{et al.} \cite{ma2026birdcast} developed Birdcast, which maximizes network-wide utility by jointly optimizing bird's-eye-view (BEV) feature selection and multicast grouping.

Recent studies have incorporated AoI metrics to manage data freshness in CP systems \cite{fang2026aoi, fang2025r, han2026spatiotemporal}. To mitigate feature misalignment caused by spatiotemporal heterogeneity, Han \textit{et al.} \cite{han2026spatiotemporal} proposed a fusion framework that exploits network synchronization and AoI to compensate for clock drifts and communication delays. Wu \textit{et al.} \cite{wu2025fresh2comm} developed an AoI-driven optimization framework that jointly controls computing and communication delay. Fang \textit{et al.} \cite{fang2025r} formulated an Age of Perceived Targets (AoPT) minimization problem to prioritize high-quality and task-relevant data for critical targets. To handle asynchronous multi-source updates in CP systems, Wang \textit{et al.} \cite{wang2023age} designed a scheduling policy that minimizes channel utilization while satisfying AoI constraints. Zhu \textit{et al.} \cite{zhu2026timeliness} introduced a timeliness-aware prioritized scheduling algorithm for multi-region CP systems, leveraging Lyapunov optimization to balance AoI reduction and communication costs. However, the aforementioned works have not incorporated users’ local sensing capabilities for AoI minimization, which are salient characteristics of CP systems.

\begin{figure}
    \centering
    \includegraphics[width=0.92\linewidth]{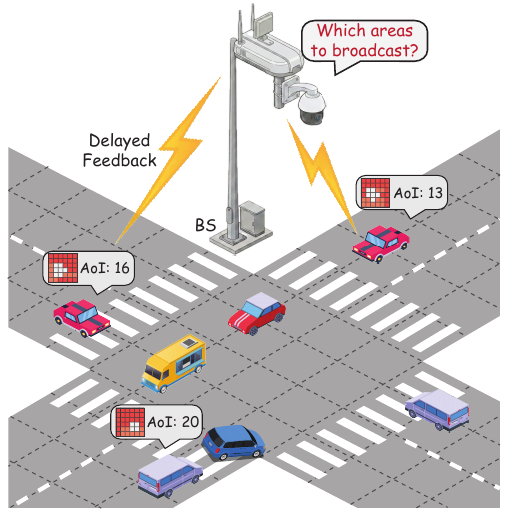}
    \caption{Illustration of our framework for an infrastructure-assisted CP system. Users, i.e., vehicles, perceive their immediate surroundings via onboard sensors, instantaneously resetting their AoI to 0. Operating on \textit{delayed} feedback, the BS dynamically schedules broadcast updates for up to $K$ areas per slot to complement the users' ongoing local sensing and minimize the network's total AoI.}
    \label{fig:system}
    \vspace{-0.0cm}
\end{figure}

\section{System Model and Problem Formulation} \label{sec:system_model}
\subsection{System Model}
We consider a real-time CP system comprising a BS and a dynamic set of users (e.g., vehicles), as illustrated in Fig.~\ref{fig:system}. The BS coverage region is partitioned into $B$ distinct areas, denoted by the set $\mathcal{B}=\{1,2,\ldots, B\}$, and time is discretized into a slotted horizon $\mathcal{T}=\{1,2,\ldots, T\}$. Both the BS and users are equipped with sensing capabilities (e.g., cameras or LiDAR) to perceive the environment. We assume that the BS continuously monitors the entire region of interest, whereas each user observes only a \textit{time-varying} subset of the region due to mobility and physical occlusions. By adopting appropriate modulation and coding schemes, we assume a broadcast update from the BS can reach all users interested in the corresponding area~\cite{ma2026birdcast}. As illustrated in Fig. \ref{fig:aoi_evolution}, the AoI evolution of user $i$ for area $b$ at time slot $t$, $a_{b,i}(t)$, comprises the following cases.

\begin{itemize}
\item \textit{Decreases to 1}: If user $i$ cannot directly observe area $b$ but receives a broadcast update from BS at slot $t$, its AoI for area $b$ becomes $1$ at slot $t+1$.

\item \textit{Resets to 0}: If user $i$ directly observes area $b$ or loses interest in area $b$, the AoI is set to $0$ at slot $t+1$.

\item \textit{Increases by 1}: Otherwise, the AoI increases by 1 in each time slot.
\end{itemize}

\begin{figure}
    \centering
    \includegraphics[width=0.92\linewidth]{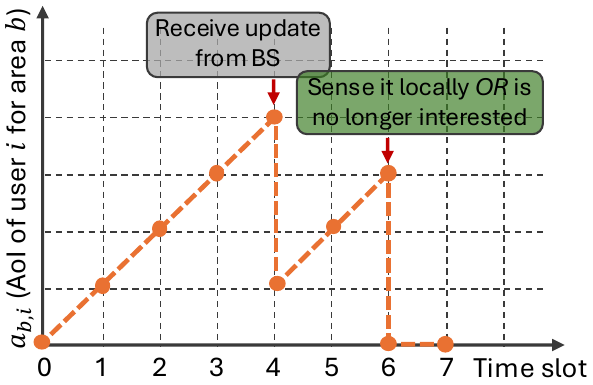}
    \caption{Evolution of the AoI $a_{b,i}$ of user $i$ for area $b$ over time slots. The AoI increases by 1 in each time slot if the user is interested in area $b$ but cannot sense it. AoI reduces to 1 upon receiving a new update from the BS and resets to 0 if the user senses the area or loses interest in it.}
    \label{fig:aoi_evolution}
    \vspace{-0.0cm}
\end{figure}

\textbf{User Demand Modeling.} AoI evolution depends on user demands, i.e., when they can see the area and when they lose interest. Let $N_{b}(t)$ denote the instantaneous user demand, defined as the number of users in the network \textit{unable} to directly observe but \textit{interested in} area $b$ at time $t$. $N_{b}(t)$ changes according to user mobility and environmental dynamics. To capture the temporal correlation and the non-negative integer nature of this population, we model $N_{b}(t)$ as a first-order integer-valued auto-regressive (INAR(1)) process \cite{al1987first, brijs2008studying}:
\begin{equation}
    N_{b}(t)=\rho_{b}\circ_{c}N_{b}(t-1)+\omega_{b}(t), \label{eq:AR}
\end{equation}
where $\rho_{b} = \rho_{b}^{int}(1-\rho_{b}^{vis}) \in(0,1)$ represents the probability that a user remains in the interested-but-unobserving state for area $b$. Here, $\rho_{b}^{int}$ and $\rho_{b}^{vis}$ denote the probabilities of the user maintaining interest and possessing direct observation of area $b$, respectively. 
The symbol $\circ_{c}$ is a dependent binomial thinning operator, \newrev{defined as $\rho_{b} \circ_{c}N_{b}(t-1) = \sum_{i=1}^{N_{b}(t-1)}X_{i}(t-1)$ with $X_{i}(t-1)\sim \text{Bernoulli}(\rho_{b})$ being identically distributed but potentially correlated Bernoulli random variables} (this correlation accounts for the spatial dependencies typically exhibited by proximate users), and $\omega_{b}(t)$ is an independent, integer-valued random variable representing the influx of users newly transitioning into the interested-but-unobserving state for area $b$ at slot $t$ due to mobility or physical occlusions.
Considering physical road capacity constraints, there exist finite constants $N_{\max}$ and $\mu_{\max}$ such that $N_{b}(t)\le N_{\max}$ and $\omega_{b}(t) \le \mu_{\max}$ for all $b \in \mathcal{B}$.


\textbf{Delayed Status Upload.} To track information staleness, each user reports its location, a perception-quality or visibility indicator\footnote{This indicator can be directly derived from the user's perception confidence \cite{hu2022where2comm}, which is generated by onboard perception modules without incurring extra computing overhead.}, and its local AoI for each target area to the BS. Due to uplink latency, however, the BS observes these reports with a constant delay of $d \geq 1$ time slots and must operate on delayed observations when making scheduling decisions.
Accordingly, the information available to the BS at the beginning of slot $t$ is $\mathcal{H}(t)=\{N_{b}(\tau), A_{b}(\tau): b\in\mathcal{B}, \tau\leq t-d\}$. The BS constructs $N_b(\tau)$ by counting the number of users that are interested in, but unable to directly observe area $b$ and aggregating their individual AoI reports, such that $A_{b}(\tau)= \sum_{i=1}^{N_{b}(\tau)} a_{b,i}(\tau)$. The evolution of $a_{b,i}$ will be detailed in Section \ref{sec:aoi_def}. 
Note that $N_{b}(\tau)$ and $A_{b}(\tau)$ are derived from uploaded binary matrices and scalar values from users, respectively, which are communication-efficient.
The key notations are summarized in Table \ref{table:notations}.

\begin{table}[!t]
\centering
\caption{Summary of important notations.}
\label{table:notations}
\renewcommand{\arraystretch}{1.4}
\setlength{\tabcolsep}{2mm}
\begin{tabular}{@{}p{0.9cm}p{7.0cm}@{}}
\hline
\textbf{Notation} & \textbf{Description} \\ 
\hline
$\mathcal{T}$ & The set of time slots \\
$\mathcal{B}$ & The set of distinct areas \\
$N_{b}(t)$ & The instantaneous user demand for area $b$ at slot $t$\\ 
$\rho_{b}$ & The probability of a user remaining in the interested-but-unobservable state for area $b$\\
$\omega_{b}(t)$  & The number of new interested-but-unobservable users for area $b$ at slot $t$ \\
$\lambda_{b}$ & The steady-state mean user demand for area $b$\\
$\mu_b$ & The mean arrival rate of interested-but-unobservable users for area $b$\\ 
$d$  & The observation delay (in time slots) \\
$\mathcal{H}(t)$ & The delayed information available to the BS  \\
$a_{b,i}(t)$ & AoI of user $i$ for area $b$ at slot $t$ \\
$A_b(t)$ & The aggregate AoI for area $b$ at slot $t$\\
$u_b(t)$ & The binary scheduling decision for area $b$ at slot $t$\\
$K$ & The maximum number of areas can be scheduled for updates \\
$\overline A_b$ & The long-term time-average AoI for area $b$ \\
$\hat{N}_b(t)$ & The predicted user demand for area $b$ at slot $t$ \\
$\hat{A}_b(t)$ & The predicted aggregate AoI of area $b$ at slot $t$\\
\hline
\end{tabular}
\end{table}

\subsection{Problem Formulation}
At the beginning of each time slot $t$, the BS generates fresh information for all $B$ areas via its equipped sensors and then selects \textit{a subset of areas to broadcast}. Let $u_b(t)\in\{0,1\}$ denote the binary scheduling variable, where $u_b(t)=1$ if the BS broadcasts an update for area $b$ in slot $t$, and $u_b(t)=0$ otherwise. Our objective is to find the optimal scheduling policy $\pi^*$, i.e., the sequence of scheduling actions $\{u_b(t)\}_{b \in \mathcal{B}}$, that minimizes the long-term time-average sum AoI across the network subject to a bandwidth constraint. By defining $\overline A_b = \limsup_{T\to\infty} \frac{1}{T} \sum_{t=1}^{T} \mathbb{E} \left[ A_b(t) \right]$, the corresponding optimization problem is formulated as
\begin{subequations}\label{p:min_AoI}
\begin{align}
    \min_{\pi}~~ & \sum_{b=1}^B \overline A_b \label{eq:obj} \\
    \text{s.t.} ~~ & \sum_{b=1}^{B} u_b(t) \le K, \quad \forall t \ge 1, \label{eq:c1} \\
    & u_b(t) \in \{0,1\}, \quad \forall b \in \mathcal{B}, ~\forall t \ge 1, \label{eq:c2}
\end{align}
\end{subequations}
where Constraint \eqref{eq:c1} indicates that at most $K$ areas can be scheduled for broadcast in any single time slot.

As previously discussed, the true system state, specifically, the current demand and aggregate AoI, $\{N_{b}(t), A_{b}(t)\}_{b\in\mathcal{B}}$, is not observable by the BS when making decisions. Instead, the BS must rely on delayed information and historical actions to schedule updates. This naturally leads to a partially observable restless multi-armed bandit (PO-RMAB) problem. In general, PO-RMAB problems are PSPACE-hard, and determining an optimal policy entails exponential complexity and memory requirements \cite{mundhenk2000complexity, shao2021partially}.

Moreover, different from traditional AoI scheduling problems, even evaluating the objective function \eqref{eq:obj} is challenging in our case. \ding{182} First, information freshness in CP is intertwined with users’ time-varying perception ranges: whenever a user can directly sense a particular area, the AoI of that area is effectively reset to zero even \textit{without} receiving an update. \ding{183} Second, the user population in the network \textit{evolves over time}, which further complicates the analysis of the time-average sum AoI. Consequently, before solving problem \eqref{p:min_AoI}, we must first develop a tractable expression for the objective \eqref{eq:obj}.

\section{Characterization of Time-average Sum AoI}
\label{sec:aoi_def}
In this section, we derive a tractable expression for the time-average sum AoI in \eqref{eq:obj}, which forms the basis for the subsequent development of our solution. Based on this expression, we transform the original AoI minimization into an equivalent weighted AoI reduction maximization problem.

\subsection{Characterization of Time-average AoI}
Recall that $a_{b,i}(t)$ denotes the AoI of user $i$ for area $b$ at slot $t$\footnote{The index $i$ identifies users only within slot $t$ and is not tracked consistently over time.}, and $A_b(t)=\sum_{i=1}^{N_b(t)} a_{b,i}(t)$ represents the instantaneous aggregate AoI of all interested-but-unobserving users. Since our objective is to minimize the time-average sum AoI $\overline A_b$, we will show that it is sufficient to directly characterize $A_b(t)$, without explicitly tracking the evolution of each individual $a_{b,i}(t)$. Given the scheduling decision $u_b(t)$, the aggregate AoI $A_{b}(t)$ evolves as
\begin{equation}
    A_{b}(t+1)= \underbrace{\sum_{i=1}^{N_{b}(t)}X_i(t)\big[(1-u_b(t))a_{b,i}(t)+1\big]}_{\text{staying users}} + \underbrace{\omega_{b}(t+1)}_{\text{new arrivals}}. \label{eq:aoi_def}
\end{equation}
It comprises two components based on user population dynamics. \ding{182} The first term accounts for \textit{staying users}, who remain in the interested-but-unobserving state. The indicator $X_i(t) \sim \text{Bernoulli}(\rho_{b})$ denotes whether user $i$ persists in this state at slot $t+1$. For these users, if BS schedules an update, i.e., $u_b(t)=1$, their AoI resets to $1$ at slot $t+1$; otherwise, it increments by $1$. \ding{183} The second term accounts for \textit{new arrivals}, representing a batch of $\omega_{b}(t+1)$ users newly transitioning into this state. Since these users could either observe area $b$ locally at slot $t$ or have just entered the network, their initial AoI at slot $t$ is $0$.


Let $\lambda_b \triangleq \lim_{t \to \infty} \mathbb{E}[N_b(t)]$ denote the steady-state mean demand for area $b$ and let $\mu_b \triangleq \mathbb{E}[\omega_b(t)]$ denote the mean influx of users, where $\mu_b \le \mu_{\max}$. Under the INAR(1) mobility model, we have $\mu_b = (1-\rho_b)\lambda_b$. We now establish the one-step evolution and the stability of $A_b(t)$ and then characterize the tractable expression of the time-average AoI $\overline{A}_{b}$.

\begin{lemma}\label{lemma:one_step_drift}
   Given the system filtration $\mathcal{F}_t$ containing all system states up to time slot $t$, the conditional one-step evolution of the aggregate AoI for area $b$ is given by
   \begin{equation} \label{eq:one_step_aoi}  
        \mathbb{E}[A_b(t+1) | \mathcal{F}_t] = \rho_b \big(1 - u_b(t)\big) A_b(t) + \rho_b N_b(t) + \mu_b. 
    \end{equation}  
\end{lemma}

\begin{IEEEproof}
The proof can be found in Appendix \ref{proof:one_step_drift}
\end{IEEEproof}

\begin{lemma}\label{lemma:stability}
    Under any admissible scheduling policy, the aggregate AoI \(A_b(t)\) is mean-rate stable, i.e., $\lim_{T\to\infty}\frac{\mathbb{E}[A_b(T+1)]}{T}=0$, which implies 
    \begin{equation}
        \lim_{T\to\infty}\frac{1}{T}\Big(\mathbb{E}[A_b(T+1)]-\mathbb{E}[A_b(1)]\Big) =0.
    \end{equation}
\end{lemma}

\begin{IEEEproof}
 The proof can be found in Appendix \ref{proof:stability}.
\end{IEEEproof}

\begin{proposition}\label{prop:characterization} 
Based on Lemma \ref{lemma:one_step_drift} and \ref{lemma:stability}, the time-average aggregate AoI for area $b$ is characterized by
    \begin{equation} \label{eq:prop2}
        \overline{A}_{b}=\frac{\lambda_{b}}{1-\rho_{b}}-\frac{\rho_{b}}{1-\rho_{b}}\overline{u_{b}A_{b}},
    \end{equation}
where $\overline{u_bA_b}\triangleq\lim_{T\to\infty}\frac{1}{T}\sum_{t=1}^{T} \mathbb{E}\left[u_b(t)A_b(t)\right]$.
\end{proposition}

$$\bar A_{b,T}=\frac1T\sum_{t=1}^T \mathbb E[A_b(t)],\quad
\bar U_{b,T}=\frac1T\sum_{t=1}^T \mathbb E[u_b(t)A_b(t)].
$$

\begin{IEEEproof} 
   \rev{Applying the law of iterated expectations to the one-step conditional drift established in Lemma \ref{lemma:one_step_drift} yields}
    \begin{equation}
        \mathbb{E}[A_{b}(t+1)]=\rho_{b}\mathbb{E}[A_{b}(t)]-\rho_{b}\mathbb{E}[u_{b}(t)A_{b}(t)]+\rho_{b}\mathbb{E}[N_{b}(t)]+\mu_{b}.
    \end{equation}
    Subtracting $\mathbb{E}[A_b(t)]$ from both sides and average over the time horizon from $t=1$ to $T$, we obtain
    \begin{equation}
    \begin{aligned}
        \frac{1}{T}\big(\mathbb{E}&[A_{b}(T+1)]-\mathbb{E}[A_{b}(1)]\big)=(\rho_{b}-1)\frac{1}{T}\sum_{t=1}^{T}\mathbb{E}[A_{b}(t)]\\
        &-\rho_{b}\frac{1}{T}\sum_{t=1}^{T}\mathbb{E}[u_{b}(t)A_{b}(t)]+\rho_{b}\frac{1}{T}\sum_{t=1}^{T}\mathbb{E}[N_{b}(t)]+\mu_{b}.
    \end{aligned}
    \end{equation}
    Taking the limit superior as $T \to \infty$, Lemma \ref{lemma:stability} guarantees that the boundary difference on the left-hand side vanishes. Concurrently, applying the definition of the long-term time average to the right-hand side yields
    \begin{equation}
        0 = (\rho_{b}-1)\overline{A}_{b} - \rho_{b}\overline{u_{b}A_{b}} + \rho_b \lambda_b + \mu_b.
    \end{equation}
    Substituting $\mu_{b} = (1-\rho_{b})\lambda_{b}$ simplifies the equation to
    \begin{equation}
        0 = (\rho_{b}-1)\overline{A}_{b} - \rho_{b}\overline{u_{b}A_{b}} + \lambda_{b}.
    \end{equation}
    Rearranging these terms to solve for $\overline{A}_{b}$ yields \eqref{eq:prop2}, which completes the proof.
\end{IEEEproof}

\textbf{Insights:}
\textit{The expression for $\overline{A}_{b}$ in \eqref{eq:prop2} contains two components: \ding{182} A policy-independent baseline ($\frac{\lambda_{b}}{1-\rho_{b}}$), which represents the accumulation of information staleness in area $b$ only depending on exogenous factors, i.e., $\lambda_b$ and $\rho_b$. This term acts as the upper bound on the time-average AoI when there are no updates from the BS. \ding{183} A policy-dependent term ($\frac{\rho_{b}}{1-\rho_{b}}\overline{u_{b}A_{b}}$), which captures the AoI reduction achieved via BS updates and is weighted by the coefficient $\frac{\rho_{b}}{1-\rho_{b}}$ related to vehicular mobility/occlusion factors.
}

\subsection{Problem Transformation}
Minimizing the time-average sum AoI $\sum_{b=1}^{B}\overline{A}_{b}$ in \eqref{eq:obj} is mathematically equivalent to maximizing the sum of the weighted AoI reductions. According to Proposition \ref{prop:characterization}, the optimization problem can thus be reformulated as follows
\begin{subequations}\label{p:reformulated_min_AoI}
\begin{align}
    \max_{\pi}~~ & \sum_{b=1}^{B}\frac{\rho_{b}}{1-\rho_{b}}\overline{u_{b}A_{b}} \label{eq:re_obj} \\
    \text{s.t.} ~~ & \eqref{eq:c1}, \eqref{eq:c2}.
\end{align}
\end{subequations}
However, solving the reformulated problem in \eqref{p:reformulated_min_AoI} remains PSPACE-hard. In our optimization, this difficulty is further exacerbated by two main factors. First, due to the $d$-slot observation delay, the current system state, $\{N_{b}(t), A_{b}(t)\}_{b\in\mathcal{B}}$, is \emph{hidden} from the BS at the moment the scheduling decision $u_b(t)$ must be made. Second, the underlying system dynamics,  $\boldsymbol{\theta}_b = [\rho_b, \mu_b]^\top$, are \emph{unknown} a priori and must be learned online. 
To address these coupled challenges, we propose an efficient online LocMW policy in Section \ref{sec:online_learning}. To establish performance benchmarks, we first derive a fundamental lower bound and a stationary randomized scheduling policy in Sections \ref{sec:lower_bound} and \ref{sec:randomized_policy}.

\section{Lower Bound on Time-average Sum AoI}
\label{sec:lower_bound}

In this section, we derive a mean-field lower bound on the time-average sum AoI in \eqref{p:min_AoI} for any admissible scheduling policy.

We define $p_b \triangleq \lim_{T\to\infty}\frac{1}{T}\sum_{t=1}^T\mathbb{E}[u_b(t)]$ as the expected update rate for area $b$, which satisfies $\sum_{b=1}^{B} p_{b}\le K$. For each area $b$, let $\tilde{N}_b(t) \triangleq \mathbb{E}[N_b(t) | \mathcal{H}(t)] = \lambda_b + \rho_b^d(N_b(t-d) - \lambda_b)$ denote the $d$-slot demand predictor. The corresponding prediction error is defined as $e_b(t) \triangleq N_b(t) - \tilde{N}_b(t)$, which satisfies $\mathbb{E}[e_b(t) | \mathcal{H}(t)] = 0$. 
Let $Z_{b}(t) \triangleq \mathbb{E}[A_{b}(t) | \mathcal{H}(t)]$ denote the conditional expected aggregate AoI.
By the law of iterated expectations, we have
\begin{equation}
    \begin{aligned}
        \mathbb{E}[\tilde N_b(t)] &= \mathbb{E}[\mathbb{E}[N_b(t) | \mathcal H(t)]] =\mathbb{E}[N_b(t)], \\
        \mathbb{E}[Z_b(t)] &=\mathbb{E}[\mathbb{E}[A_b(t)|\mathcal H(t)]]=\mathbb{E}[A_b(t)].
    \end{aligned}
\end{equation}
We then define the deterministic long-term means $z_b \triangleq \lim_{T\to\infty}\frac{1}{T}\sum_{t=1}^T\mathbb{E}[Z_b(t)]$ and
$x_b \triangleq \lim_{T\to\infty}\frac{1}{T}\sum_{t=1}^T\mathbb{E}[u_b(t)Z_b(t)]$. 
In addition, let $\tilde{M}_{b}(t+1) \triangleq \rho_{b}\tilde{N}_{b}(t)+\mu_{b}$ represent the predicted one-step input to the AoI dynamics conditioned on $\mathcal{H}(t)$.
Since $\mu_b=(1-\rho_b)\lambda_b$, its asymptotic expectation is given by $\lim_{t\to\infty}\mathbb{E}[\tilde{M}_{b}(t+1)]=\rho_{b}\lambda_{b}+\mu_{b}=\lambda_{b}$.
We introduce the following assumption regarding the asymptotic behavior of the predicted dynamics.

{\begin{assumption}[Mean-field Concentration]
\label{ass:mf_decoupling} 
For each area \(b\), the predicted one-step input satisfies
\begin{equation}
\begin{aligned}
    \lim_{T\to\infty}\frac{1}{T}\sum_{t=1}^{T}
\operatorname{Var}[\tilde M_b(t+1)]=0. 
\end{aligned}
\end{equation}
\end{assumption}
Assumption~\ref{ass:mf_decoupling} imposes a mean-field concentration on the predicted input, which is assumed to converge to its mean. This is a standard assumption in stochastic network analysis, which becomes true 
as the number of entities grows large~\cite{srikant2014communication, benaim2008class}. In our system, it becomes asymptotically accurate for large unobserving populations (e.g., at dense urban intersections). Under this premise, we establish the following theorem.
}

\begin{theorem}\label{thm:lower_bound}
Under Assumption~\ref{ass:mf_decoupling}, the time-average sum AoI achieved by any admissible scheduling policy satisfies
\begin{equation}
    \sum_{b=1}^{B}\bar A_b
    \ge C^{\mathrm{LB}} \triangleq \min_{\substack{0\le p_b\le 1\\ \sum_{b=1}^{B}p_b\le K}} \sum_{b=1}^{B} L_b(p_b), \label{eq:mf_lb}
\end{equation}
where $L_b(p_b) \triangleq \lambda_b \frac{p_b+1}{p_b(1+\rho_b)+1-\rho_b}$.
Moreover, the optimal solution of \eqref{eq:mf_lb} has the water-filling form
\begin{equation}
    p_b^*(\gamma^*)=\left[\frac{\sqrt{2\lambda_b\rho_b/\gamma^*}-(1-\rho_b)}{1+\rho_b}\right]_0^1,\label{eq:lb_water_filling}
\end{equation}
where $[x]_0^1 \triangleq \min\{1, \max\{0,x\}\}$, and $\gamma^* > 0$ is a Lagrange multiplier chosen such that $\sum_{b=1}^B p_b^*(\gamma^*) = K$. The optimal value of $\gamma^*$ can be found using the bisection method.
\end{theorem}

\begin{IEEEproof}
    The proof can be found in Appendix \ref{proof:lower_bound}.
\end{IEEEproof}

Note that \(C^{\mathrm{LB}}\) is a structural lower bound based on a mean-field assumption with fractional update rates and known statistical parameters, serving as an analytical baseline for performance evaluation.

\section{Optimal Randomized Scheduling Policy}
\label{sec:randomized_policy}
This section develops an optimal randomized scheduling policy for Problem~(\ref{p:min_AoI}). 
Unlike state-aware policies, this policy depends only on the statistical parameters $\{\lambda_b,\rho_b\}_{b\in\mathcal{B}}$ and can be implemented without feedback. This randomized policy serves as a strong benchmark for our subsequent development.


Let \(\Pi_{\mathrm{R}}\) denote the class of stationary randomized policies. 
Each policy \(\pi^{\mathrm{R}}\in\Pi_{\mathrm{R}}\) is specified by a vector of marginal broadcasting probabilities $\boldsymbol{\eta}^{\mathrm{R}} =[\eta_1^{\mathrm{R}},\ldots,\eta_B^{\mathrm{R}}]^\top$, where $\eta_b^{\mathrm{R}} \triangleq \Pr\{u_b^{\mathrm{R}}(t)=1\} \in [0,1]$ with $\sum_{b=1}^B \eta_b^{\mathrm{R}} \le K$.
At each slot, the BS samples a feasible subset \(\mathcal{S}^{\mathrm{R}}(t)\subseteq \mathcal{B}\) satisfying
\(|\mathcal{S}^{\mathrm{R}}(t)|\le K\), and sets \(u_b^{\mathrm{R}}(t)=1\) if and only if \(b\in \mathcal{S}^{\mathrm{R}}(t)\). This sampling process is independent across time and the system history.

\begin{theorem} 
\label{thm:optimal_randomized}
The optimal randomized marginal probability is
\begin{equation}
    \eta_b^{\mathrm R,*}(\nu^*) =\left[\frac{\sqrt{\lambda_b\rho_b/\nu^*}-(1-\rho_b)}{\rho_b}\right]_0^1,
\label{eq:optimal_randomized_prob}
\end{equation}
where the dual variable \(\nu^*>0\) is determined via one-dimensional bisection to satisfy $\sum_{b\in\mathcal B}\eta_b^{\mathrm{R},*}(\nu^*)=K$. Consequently, the minimum expected time-average sum AoI achieved by this policy is
\begin{equation}
    C^{\mathrm{R},*}=\sum_{b=1}^B\frac{\lambda_b}{1-\rho_b+\rho_b\eta_b^{\mathrm{R},*}}.
\label{eq:optimal_randomized_value}
\end{equation}
\end{theorem}

\begin{IEEEproof}
Taking the expectation of the one-step AoI evolution in \eqref{eq:one_step_aoi}, and noting that $\mathbb{E}[u_b^{\mathrm{R}}(t)A_b(t)] = \eta_b^{\mathrm{R}} \mathbb{E}[A_b(t)]$ as $u_b^{\mathrm{R}}(t)$ is independent of the current AoI state, yields
\begin{equation}
    \mathbb E[A_b(t+1)]=\rho_b(1-\eta_b^{\mathrm{R}})\mathbb E[A_b(t)]+\rho_b\mathbb E[N_b(t)]+\mu_b .
\end{equation}
Averaging over \(t=1,\ldots,T\), letting \(T\to\infty\), and invoking the mean-rate stability established in Lemma \ref{lemma:stability}, we obtain
\begin{equation}
    0=-\bigl(1-\rho_b+\rho_b\eta_b^{\mathrm{R}}\bigr)\bar A_b^{\mathrm{R}}+\rho_b\lambda_b+\mu_b.
\end{equation}
Substituting $\mu_b = (1-\rho_b)\lambda_b$ and solving for \(\bar A_b^{\mathrm{R}}\) gives 
\begin{equation}
    \bar{A}_b^{\mathrm{R}}(\eta_b^{\mathrm{R}}) = \frac{\lambda_b}{1-\rho_b+\rho_b\eta_b^{\mathrm{R}}} \triangleq F_b(\eta_b^{\mathrm{R}}).
\end{equation}
The optimal randomized policy is thus obtained by minimizing the sum of these individual functions:
\begin{equation}
C^{\mathrm{R},*} \triangleq \min_{\boldsymbol{\eta}\in\mathcal{P}_K} \sum_{b=1}^B F_b(\eta_b),
\label{eq:randomized_optimization}
\end{equation}
where $\mathcal{P}_K \triangleq \left\{\boldsymbol{\eta}\in[0,1]^B : \sum_{b=1}^B \eta_b \le K\right\}$ is the feasible set. 
It is straightforward to verify that $F_b(\eta_b^{\mathrm{R}})$ is convex and monotonically decreasing with respect to $\eta_b^{\mathrm{R}}$. 
Let $\nu \ge 0$ be the Lagrange multiplier associated with the bandwidth constraint. Solving the KKT stationarity conditions yields
\begin{equation}
    \eta_b = \frac{\sqrt{\lambda_b\rho_b/\nu} - (1-\rho_b)}{\rho_b}.
\end{equation}
Projecting this onto $[0,1]$ results in the formulation in \eqref{eq:optimal_randomized_prob}. The optimal $\nu^*$ is efficiently found via bisection. Finally, substituting $\eta_b^{\mathrm{R},*}$ back into the objective provides the performance bound in 
\eqref{eq:optimal_randomized_value}, completing the proof.
\end{IEEEproof}

\begin{theorem}\label{thm:randomized_guarantee}
The optimal stationary randomized policy satisfies
\begin{equation}
    C^{\mathrm{LB}} \le C^{\mathrm{R},*} \le 2C^{\mathrm{LB}}.
    \label{eq:randomized_factor_two}
\end{equation}
\end{theorem}

\begin{IEEEproof}
Recall that the per-area term in the mean-field lower bound is defined as $L_b(x)\triangleq\lambda_b\frac{x+1}{x(1+\rho_b)+1-\rho_b}$. 
For any $x \in [0,1]$, evaluating the ratio between $F_b(x)$ and $L_b(x)$ yields
\begin{equation}
    \frac{F_b(x)}{L_b(x)} = \frac{1-\rho_b+(1+\rho_b)x}{(1-\rho_b+\rho_b x)(1+x)}.
\end{equation}
We have
\begin{equation}
    \begin{aligned}
        &2(1+x)(1-\rho_b+\rho_b x) - \big(1-\rho_b+(1+\rho_b)x\big) \\
    = &(1-\rho_b)(1+x) + 2\rho_b x^2 >0,
    \end{aligned}
\end{equation}
which implies $F_b(x) \le 2L_b(x)$ for all $b \in \mathcal{B}$ and $x \in [0,1]$.
Let $\mathbf{p}^{\mathrm{LB}} = [p_1^{\mathrm{LB}}, \ldots, p_B^{\mathrm{LB}}]^\top$ denote the optimal allocation vector for the lower-bound. As $\mathbf{p}^{\mathrm{LB}} \in \mathcal{P}_K$ is also a feasible marginal probability vector for the randomized policy, the optimal performance is bounded by
\begin{equation}
    C^{\mathrm{R},*} \le \sum_{b=1}^B F_b(p_b^{\mathrm{LB}}) \le 2\sum_{b=1}^B L_b(p_b^{\mathrm{LB}}) = 2C^{\mathrm{LB}},
\end{equation}
completing the proof.
\end{IEEEproof}

 Theorem~\ref{thm:randomized_guarantee} establishes that the optimal randomized policy achieves a 2-approximation of the mean-field lower bound. Since this randomized approach does not condition on state history, its performance $C^{\mathrm{R},*}$ is independent of the observation delay $d$.

\section{Local-sensing-aware Max-Weight Policy}
\label{sec:online_learning}

In this section, we develop an efficient online learning approach, termed Local-sensing-aware Max-Weight (LocMW) scheduling, to solve the PO-RMAB problem formulated in Sections~\ref{sec:system_model} and \ref{sec:aoi_def}. We first detail the proposed LocMW policy and then derive theoretical performance guarantees.

\subsection{The LocMW Framework}
The proposed LocMW policy employs a three-step learning-and-control framework. Specifically, the BS first estimates the system parameters $\boldsymbol{\theta}_b = [\rho_b, \mu_b]^\top$ from delayed reports, then predicts the current network states, and finally schedules updates via a Max-Weight scheduling approach.

\textbf{\textit{Step 1. Online Parameter Estimation:}}
At the beginning of slot $t$, the BS observes the delayed historical states \(\mathcal H(t)=\{N_b(\tau),A_b(\tau): b\in\mathcal B,\tau\le t-d\}\). For each area \(b\), define the state regressor $\mathbf{x}_b(\tau) = [N_b(\tau-1), 1]^\top$ and the target observation $y_b(\tau) = N_b(\tau)$.
Under the INAR(1) demand model, the expected demand evolves as $\mathbb E[y_b(\tau)\mid \mathbf{x}_b(\tau)]= \mathbf{x}_b(\tau)^\top\boldsymbol{\theta}_b$.
Thus, the BS estimates \(\boldsymbol{\theta}_b\) via an $\ell_2$-regularized least squares (ridge regression) estimator \cite{hastie2009elements, zhang2010regularized}:
\begin{equation}
    \tilde{\boldsymbol{\theta}}_b(t) = \arg\min_{\boldsymbol{\theta} \in \mathbb{R}^2} \left( \sum_{\tau=2}^{t-d} \left( y_b(\tau) - \mathbf{x}_b(\tau)^\top \boldsymbol{\theta} \right)^2 + \zeta\|\boldsymbol{\theta}\|_2^2 \right). \label{eq:theta_est}
\end{equation} 
where $\zeta >0$ is the regularization parameter. 
The $\ell_2$-penalty prevents the parameter estimates from overfitting to stochastic traffic anomalies.

To enhance computational efficiency, this estimator is computed recursively rather than from scratch at every slot. By maintaining $\mathbf{G}_b(t)=\zeta\mathbf{I}+\sum_{\tau=2}^{t-d} \mathbf{x}_b(\tau)\mathbf{x}_b(\tau)^\top$ and $\mathbf{h}_b(t) = \sum_{\tau=2}^{t-d}\mathbf{x}_b(\tau)y_b(\tau)$, the estimate simplifies to 
\begin{equation}
    \tilde{\boldsymbol{\theta}}_b(t) = \mathbf{G}_b(t)^{-1}\mathbf{h}_b(t). \label{eq:theta_Gh}
    \end{equation}
Since $\mathbf{G}_b(t) \in \mathbb{R}^{2\times2}$, its inversion incurs negligible $\mathcal{O}(1)$ complexity. In slot $t+1$, $\mathbf{G}_b$ and $\mathbf{h}_b$ are incrementally updated with the newly available observation $(\mathbf{x}_b(t+1-d), y_b(t+1-d))$.

Finally, to ensure system stability, this estimate $\tilde{\boldsymbol{\theta}}_b(t)$ is projected onto the parameter bounds to yield the final result
\begin{equation} \label{eq:theta_proj}
   \hat{\boldsymbol{\theta}}_b(t) = [\hat{\rho}_b(t), \hat{\mu}_b(t) ]^\top = \operatorname{Proj}_{\Theta} \big( \tilde{\boldsymbol{\theta}}_b(t) \big),
\end{equation} 
where $\operatorname{Proj}_{\Theta}(\cdot)$ denotes the Euclidean projection operator onto the compact set $\Theta \triangleq [0, \rho_{\max}] \times [0, \mu_{\max}]$ with $\rho_{\max} < 1$.

\textbf{\textit{Step 2. Certainty-Equivalent State Prediction:}}
To mitigate the intractability of exact belief-state tracking over the unobservable delay window, we adopt the certainty equivalence principle. Specifically, the BS predicts the current system state. Starting from the most recently received delayed reports, $\hat{N}_b(t-d) = N_b(t-d)$ and $\hat{A}_b(t-d) = A_b(t-d)$, the BS recursively propagates the estimated dynamics over the delay window $\tau = t-d, \dots, t-1$, according to the following rules:
\begin{equation}
\begin{aligned} \label{eq:state_pred}
    &\hat{N}_b(\tau+1) = \hat{\rho}_b(t) \hat{N}_b(\tau) + \hat{\mu}_b(t), \\
    &\hat{A}_b(\tau+1) = \hat{\rho}_b(t) \big(1 - u_b(\tau)\big) \hat{A}_b(\tau) + \hat{\rho}_b(t) \hat{N}_b(\tau) + \hat{\mu}_b(t).    
\end{aligned}
\end{equation} 
This recursive process successfully bridges the observation gap, yielding the current demand $\hat{N}_b(t)$ and AoI $\hat{A}_b(t)$ required for executing scheduling decisions.

\textbf{\textit{Step 3. Local-sensing-aware Max-Weight Scheduling:}}
Given the estimated parameters $\hat{\boldsymbol{\theta}}_b(t)=[\hat{\rho}_b(t),\hat{\mu}_b(t)]^\top$ from Step 1 and the predicted states $\hat{N}_b(t)$ and $\hat{A}_b(t)$ from Step 2, the BS determines which areas to update subject to the bandwidth constraint \eqref{eq:c1}. We now derive the scheduling index from a one-slot Lyapunov drift~\cite{kadota2018scheduling, raghunathan2008index, zhao2025optimizing}. 

We first define a linear Lyapunov function as
\begin{equation}
    L(t)\;\triangleq\;\sum_{b=1}^{B}\beta_b\,A_b(t),
\label{eq:mw-lyap}
\end{equation}
where $\beta_b > 0$ is a tunable, area-specific parameter.
Based on the estimated AoI dynamics, the one-slot Lyapunov drift is
\begin{equation}
\begin{aligned}
    \Delta(t)  \triangleq & ~\mathbb{E}[L(t+1)-L(t)\mid\mathcal{H}(t)] \\
    =&\sum_{b=1}^{B}\beta_b\Bigl[(\hat{\rho}_b(t)-1)\hat{A}_b(t)+\hat{\rho}_b(t)\hat{N}_b(t)+\hat{\mu}_b(t)\Bigr]\\
    &-\sum_{b=1}^{B}\beta_b\hat{\rho}_b(t)\hat{A}_b(t)u_b(t). \label{eq:mw-drift}    
\end{aligned}
\end{equation}
Therefore, minimizing the drift is equivalent to maximizing $\sum_b \beta_b\hat{\rho}_b(t)\hat{A}_b(t)u_b(t)$,
which can be achieved by selecting the $K$ areas with the largest weights
\begin{equation}
    W_b(t)=\beta_b\hat{\rho}_b(t)\hat{A}_b(t).
\label{eq:mw-index-generic}
\end{equation}
We define the weighting coefficient by aligning with the stationary randomized benchmark in Section \ref{sec:randomized_policy}
\begin{equation}
    \beta_b \triangleq \frac{1}{c_b},\qquad 
    c_b\triangleq 1-\hat{\rho}_b(t)+\hat{\rho}_b(t)\hat\eta_b^{\mathrm{R}}(t),
    \label{eq:mw-weights}
\end{equation}
where $\eta_b^{\mathrm{R}}(t)$ is obtained by water-filling in \eqref{eq:optimal_randomized_prob} with $\hat{\boldsymbol{\theta}}_b(t)$.
Substituting this into \eqref{eq:mw-index-generic} gives the proposed LocMW index
\begin{equation}
    W_b^{\mathrm{L}}(t)= \frac{\hat{\rho}_b(t) \hat{A}_b(t)}{1-\hat{\rho}_b(t)+\hat{\rho}_b(t) \hat\eta_b^{\mathrm R}(t)}. \label{eq:mw-index}
\end{equation}
At each slot $t$, the BS deterministically schedules the subset $\mathcal{K}^{\mathrm{L}}(t)$ containing the $K$ areas with the largest values of \eqref{eq:mw-index}. The scheduling decision is therefore
\begin{equation} \label{eq:final_decision}
    u_b^{\mathrm{L}}(t)=\begin{cases}
    1, & b\in \mathcal{K}^{\mathrm{L}}(t),\\
    0, & \text{otherwise}.
\end{cases}    
\end{equation}
As a result, the proposed LocMW algorithm dynamically prioritizes limited communication resources based on information staleness, underlying local perception capabilities, and user mobility characteristics. 
The overall procedure is outlined in Algorithm \ref{alg:online_learning}. 

\begin{remark}
    The LocMW index is local-sensing aware because the factor \(\hat\rho_b(t)\) captures the estimated persistence of users in the interested-but-unobserving state. A smaller \(\hat\rho_b(t)\) indicates that users in area \(b\) are highly likely to regain local visibility or lose interest, reducing the expected benefit of a BS update. 
\end{remark}

\noindent \textbf{Computational Complexity.}
At each time slot, parameter estimation across $B$ areas incurs a complexity of $\mathcal{O}(B)$, the $d$-step state prediction requires $\mathcal{O}(dB)$ operations, and executing the Max-Weight scheduling takes $\mathcal{O}(B\log B)$. Consequently, the overall worst-case per-slot computational complexity is $\mathcal{O}(B\log B+dB)$. This demonstrates that the proposed LocMW policy is highly scalable and computationally efficient, making it well-suited for time-sensitive applications such as CP.

\begin{algorithm}[t]
	\caption{The LocMW Scheduling Policy}
	\label{alg:online_learning}
	\LinesNumbered 
	\KwIn{$\mathcal{B}$, $K$, $d$, $\Theta$, $\zeta$}
	\KwOut{$u_b^{\mathrm{L}}(t), \forall b \in \mathcal{B}$}
    
Initialize $\mathbf{G}_b\leftarrow \zeta \mathbf{I}_2$ and $\mathbf{h}_b\leftarrow \mathbf{0}$ for all $b\in\mathcal{B}$\;

\For{$t=1,2,\ldots$}{
    \If{$t\le d$}{
        Use any warm-up schedule satisfying \eqref{eq:c1}\;
        \textbf{continue}\;
    }

    \tcp{Parameter Estimation}
    \For{$b\in\mathcal{B}$}{
        Construct $\mathbf{x}_b(t-d)=[N_b(t-d-1),1]^\top$ and $y_b(t-d)=N_b(t-d)$\;
        Update $\mathbf{G}_b\leftarrow \mathbf{G}_b+\mathbf{x}_b(t-d)\mathbf{x}_b(t-d)^\top$ and
        $\mathbf{h}_b\leftarrow \mathbf{h}_b+\mathbf{x}_b(t-d)y_b(t-d)$\;
        Compute $\hat{\boldsymbol{\theta}}_b(t)$ based on \eqref{eq:theta_Gh} and \eqref{eq:theta_proj}\;
    }

    Compute $\hat{\eta}_b^\mathrm{R}(t)$ by the water-filling rule in \eqref{eq:optimal_randomized_prob} using
    $\hat{\boldsymbol{\theta}}_b(t)$\;

    \tcp{State Prediction}
    \For{$b\in\mathcal{B}$}{
        Initialize $\hat{N}_b(t-d)\leftarrow N_b(t-d)$ and $\hat{A}_b(t-d)\leftarrow A_b(t-d)$\;
        Predict $\hat{N}_b(t)$ and $\hat{A}_b(t)$ via the recursion \eqref{eq:state_pred}\;

        \tcp{LocMW Index Derivation}
        Compute $\hat{c}_b^\mathrm{R}(t)\leftarrow
        1-\hat{\rho}_b(t)+\hat{\rho}_b(t)\hat{\eta}_b^\mathrm{R}(t)$\;
        Derive the LocMW index $W_b^{\mathrm{L}}(t)$ via \eqref{eq:mw-index}\;
    }

    Select $\mathcal{K}^{\mathrm{L}}(t)$ as the $K$ areas with the largest
    $W_b^{\mathrm{L}}(t)$\;
    Obtain the scheduling decision $u_b^{\mathrm{L}}(t)$ via \eqref{eq:final_decision}\;
    }
 \Return{${u}_b(t), \forall b \in \mathcal{B}$}
\end{algorithm}

\subsection{Performance Guarantee}
We next establish the performance guarantees of LocMW. We first consider the setting in which the parameters $(\rho_b,\mu_b), \forall b\in\mathcal{B}$ are known. In this case, the BS can exactly compute the conditional predictors $Z_b(t)\triangleq \mathbb E[A_b(t)| \mathcal{H}(t)], \tilde{N}_b(t)\triangleq \mathbb E[N_b(t)|\mathcal{H}(t)]$.


\begin{theorem}
    Suppose that \((\rho_b,\mu_b)\) are known for all \(b\in\mathcal B\) and LocMW schedules the $K$ areas with the largest weights $W_b(t)=\frac{\rho_b Z_b(t)}{1-\rho_b+\rho_b\eta_b^{\mathrm{R},*}}$, then its time-average sum AoI $\sum_{b\in\mathcal B}\bar A_b^{\mathrm{L}}$ satisfies
    \begin{equation}
        C^{\mathrm{LB}} \le \sum_{b\in\mathcal B}\bar A_b^{\mathrm{L}}\le C^{\mathrm R,*} \le 2C^{\mathrm{LB}} .
    \end{equation}
\end{theorem}

\begin{IEEEproof}
    By design, the LocMW policy maximizes the last term in \eqref{eq:mw-drift}.
    Since the optimal randomized marginal vector
    \(\boldsymbol{\eta}^{\mathrm{R},*}\in[0,1]^B\) with
    \(\sum_b\eta_b^{\mathrm{R},*}\le K\) lies within the convex hull of the feasible scheduling set, we have
    \begin{equation}
    \sum_{b=1}^B \beta_b\rho_b Z_b(t)u_b^{\mathrm{L}}(t) \ge \sum_{b=1}^B \beta_b\rho_b Z_b(t)\eta_b^{\mathrm{R},*}.    
    \end{equation}
    Therefore, the one-slot drift satisfies
    \begin{equation}
       \Delta(t)\le\sum_{b=1}^B\beta_b\bigl[-c_b^{\mathrm{R},*}Z_b(t)+\rho_b\tilde{N}_b(t)+\mu_b\bigr], 
    \end{equation}
    where $c_b^{\mathrm{R},*} \triangleq 1 - \rho_b + \rho_b \eta_b^{\mathrm{R},*}$ denotes the ideal randomized benchmark decay factor with known parameters.
    \rev{Applying the law of iterated expectations}, averaging over \(t=1,\ldots,T\), and letting \(T\to\infty\), the telescoping Lyapunov drift vanishes due to the mean-rate stability established in Lemma \ref{lemma:stability}. Applying asymptotic properties $\lim_{T\to\infty}\frac{1}{T}\sum_{t=1}^T\mathbb E[N_b(t)]=\lambda_b$ and $\mu_b=(1-\rho_b)\lambda_b$ yields
    \begin{equation}
        \sum_{b=1}^B \beta_b c_b^{\mathrm{R},*}\bar A_b^{\mathrm{L}} \le \sum_{b=1}^B \beta_b\lambda_b.
    \end{equation}
    
    Since \(\beta_b c_b^{\mathrm{R},*}=1\), this simplifies to
    \begin{equation}
        \sum_{b=1}^B \bar A_b^{\mathrm{L}} \le \sum_{b=1}^B\frac{\lambda_b}{c_b^{\mathrm{R},*}} = C^{\mathrm{R},*}.
    \end{equation}

    With the lower bound in
    Section \ref{sec:lower_bound} and the randomized-policy guarantee in Section \ref{sec:randomized_policy}, the proof is completed.  
\end{IEEEproof}

Next, we consider the case where $(\rho_b,\mu_b)$ are unknown. We introduce two fundamental conditions in the following. Assumption~\ref{assump: parameter} imposes a persistence-of-excitation condition, ensuring that the delayed demand samples exhibit sufficient temporal variation to identify $(\rho_b,\mu_b)$ \cite{lai1982least, abbasi2011improved}. Assumption~\ref{ass:locmw_index_regularity} establishes a local stochastic Lipschitz condition, which guarantees that small parameter estimation errors translate into bounded perturbations in the scheduling weights \cite{stahlbuhk2021learning}.
\begin{assumption}\label{assump: parameter}
    For each $b\in\mathcal{B}$, there exist constants $\alpha_b>0$ and $T_b^{\mathrm{PE}}<\infty$ such that, for all $t\ge T_b^{\mathrm{PE}}+d$, the empirical information matrix satisfies: 
    \begin{equation}
        \lambda_{\min}\left(\sum_{\tau=2}^{t-d}\mathbf{x}_b(\tau)\mathbf{x}_b(\tau)^\top\right) \ge \alpha_b (t-d),
    \end{equation}
   where $\lambda_{\min}(\cdot)$ denotes the minimum eigenvalue of a matrix.
\end{assumption}

Let $\boldsymbol{\theta} \triangleq [\boldsymbol{\theta}_1, \dots, \boldsymbol{\theta}_B]^\top \in \Theta^B$ denote a candidate parameter vector. 
We define the LocMW index map as $\boldsymbol{\Phi}(t; \boldsymbol{\theta}) \triangleq [\Phi_1(t; \boldsymbol{\theta}), \dots, \Phi_B(t; \boldsymbol{\theta})]^\top$ with $\Phi_b(t; \boldsymbol{\theta}) \triangleq \frac{\rho_b A_b(t; \boldsymbol{\theta}_b)}{1 - \rho_b + \rho_b \eta_b^{\mathrm{R}}(\boldsymbol{\theta})}$, where $N_b(t; \boldsymbol{\theta}_b)$ and $A_b(t; \boldsymbol{\theta}_b)$ are recursively propagated via \eqref{eq:state_pred} by substituting $\boldsymbol{\theta}_b$ for $\hat{\boldsymbol{\theta}}_b(t)$, $\eta_b^{\mathrm{R}}(\boldsymbol{\theta})$ denotes the randomized marginal probability from \eqref{eq:optimal_randomized_prob} under $\boldsymbol{\theta}$. We impose the following regularity condition.

\begin{assumption}\label{ass:locmw_index_regularity}
    There exist constants $r_{\boldsymbol{\theta}} > 0$ and $L_d^{\mathrm{L}} < \infty$, alongside non-negative $\mathcal{H}(t)$-measurable random variables $\Gamma_d(t)$ satisfying $\sup_{t \ge 1} \mathbb{E}[\Gamma_d^2(t)] \le (L_d^{\mathrm{L}})^2$.
    For any pair $\boldsymbol{\theta}, \boldsymbol{\theta}' \in \Theta^B$ satisfying $\|\boldsymbol{\theta} - \boldsymbol{\theta}^*\|_{2,\infty} \vee \|\boldsymbol{\theta}' - \boldsymbol{\theta}^*\|_{2,\infty} \le r_{\boldsymbol{\theta}}$, where $\|\boldsymbol{\theta} - \boldsymbol{\theta}'\|_{2,\infty} \triangleq \max_{b \in \mathcal{B}} \|\boldsymbol{\theta}_b - \boldsymbol{\theta}_b'\|_2$ and $\boldsymbol{\theta}^* \triangleq [\boldsymbol{\theta}_1^*, \dots, \boldsymbol{\theta}_B^*]^\top$ denotes the true parameter vector, the following local stochastic Lipschitz condition holds almost surely:
    \begin{equation}
        \big\| \boldsymbol{\Phi}(t; \boldsymbol{\theta}) - \boldsymbol{\Phi}(t; \boldsymbol{\theta}') \big\|_\infty \le \Gamma_d(t) \big\| \boldsymbol{\theta} - \boldsymbol{\theta}' \big\|_{2,\infty},
    \end{equation}
\end{assumption}
\rev{
Assumption~\ref{ass:locmw_index_regularity} establishes a local stochastic Lipschitz condition, where the random variable $\Gamma_d(t)$ bounds the sensitivity of the LocMW index to parameter estimation errors within a localized neighborhood.}

\begin{theorem}\label{thm:learned_locmw}

    Let $R_T^{\mathrm{LocMW}} \triangleq \big[ \sum_{t=1}^T \mathbb{E} \big[ \sum_{b \in \mathcal{B}} A_b(t) \big] - T C^{\mathrm{R},*} \big]^+$ denote the cumulative excess AoI of the proposed LocMW policy relative to the optimal randomized benchmark with known parameters. It satisfies
    \begin{equation}
        R_T^{\mathrm{LocMW}} = \mathcal{O}\left( K L_d^{\mathrm{L}} \sqrt{T \log(BT)} \right) + o(T).
    \end{equation} 
\end{theorem}

\begin{IEEEproof}
Define the feasible scheduling set as $\mathcal{U}_K \triangleq \big\{ \mathbf{u} \in \{0,1\}^B : \sum_{b \in \mathcal{B}} u_b \le K \big\}$. Let $\mathbf{w}^*(t)$ and $\mathbf{u}^*(t) \in \arg\max_{\mathbf{u} \in \mathcal{U}_K}\mathbf{w}^*(t)^\top \mathbf{u}$ denote the ideal Max-Weight index vector and scheduling decision under the true system parameters.
Similarly, let $\mathbf{w}^{\mathrm{L}}(t)$ and $\mathbf{u}^{\mathrm{L}}(t)$ denote the proposed LocMW index vector and scheduling decision. We define the maximum index error as $\epsilon_t \triangleq \|\mathbf{w}^{\mathrm{L}}(t) -\mathbf{w}^*(t) \|_\infty$. Since $\|\mathbf{u}\|_1 \le K$, the performance of the LocMW decision satisfies
\begin{equation}
\begin{aligned}
    \mathbf{w}^*(t)^\top \mathbf{u}^{\mathrm{L}}(t)
    &\ge \mathbf{w}^{\mathrm{L}}(t)^\top \mathbf{u}^{\mathrm{L}}(t) - K\epsilon_t \\
    &\ge \mathbf{w}^{\mathrm{L}}(t)^\top \mathbf{u}^*(t) - K\epsilon_t \\
    &\ge \mathbf{w}^*(t)^\top \mathbf{u}^*(t) - 2K\epsilon_t.
\end{aligned}
\end{equation}
Moreover, because the optimal randomized marginal probability vector lies within the convex hull of the feasible set, $\boldsymbol{\eta}^{\mathrm{R},*} \in \operatorname{conv}(\mathcal{U}_K)$, it follows that $\mathbf{w}^*(t)^\top \mathbf{u}^*(t) \ge \mathbf{w}^*(t)^\top \boldsymbol{\eta}^{\mathrm{R},*}$. Hence
\begin{equation}
\mathbf{w}^*(t)^\top \mathbf{u}^{\mathrm{L}}(t) \ge \mathbf{w}^*(t)^\top \boldsymbol{\eta}^{\mathrm{R},*} - 2K\epsilon_t.
\end{equation}

Since $\mathbf{u}^{\mathrm{L}}(t)$ is $\mathcal{H}(t)$-measurable, the conditional one-slot Lyapunov drift can be given as
\begin{equation}
\begin{aligned}
    \Delta(t)
    &= \sum_{b=1}^B \frac{(\rho_b-1)Z_b(t) + \rho_b\tilde{N}_b(t) + \mu_b}{c_b^{\mathrm{R},*}} - \mathbf{w}^*(t)^\top \mathbf{u}^{\mathrm{L}}(t) \\
    &\le \sum_{b=1}^B \frac{-c_b^{\mathrm{R},*} Z_b(t) + \rho_b\tilde{N}_b(t) + \mu_b}{c_b^{\mathrm{R},*}} + 2K\epsilon_t.
\end{aligned}
\end{equation}
\rev{Applying the law of iterated expectations}, we obtain 
\begin{equation}
    \mathbb{E}[\Delta(t)] \le - \mathbb{E}\bigg[\sum_{b=1}^B A_b(t)\bigg] + \sum_{b=1}^B \frac{\rho_b\mathbb{E}[N_b(t)] + \mu_b}{c_b^{\mathrm{R},*}} + 2K\mathbb{E}[\epsilon_t].
\end{equation}
Summing over $t=1,\ldots,T$, dividing by $T$, we obtain the upper bound on the average AoI
\begin{equation}
\begin{aligned}
    \sum_{t=1}^{T} \mathbb{E}\bigg[\sum_{b=1}^B A_b(t)\bigg] \le & ~\mathbb{E}[L(1)] - \mathbb{E}[L(T+1)] + 2K\sum_{t=1}^{T}\mathbb{E}[\epsilon_t]\\
    &+ \sum_{t=1}^{T} \sum_{b=1}^B \frac{\rho_b\mathbb{E}[N_b(t)]+\mu_b}{c_b^{\mathrm{R},*}}.   
\end{aligned}
\end{equation}
Considering $L(t) \ge 0$ and recalling the steady-state demand relation $\mu_b = (1-\rho_b)\lambda_b$, it yields:
\begin{equation}
    \sum_{t=1}^{T} \mathbb{E}\bigg[\sum_{b = 1}^B A_b(t)\bigg] - T C^{\mathrm{R},*} \le D_T + \mathbb{E}[L(1)] + 2K \sum_{t=1}^{T} \mathbb{E}[\epsilon_t],
\end{equation}
where $D_T = \sum_{b=1}^B \frac{\rho_b}{c_b^{\mathrm{R},*}} \sum_{t=1}^{T} \big( \mathbb{E}[N_b(t)] - \lambda_b \big)$ captures the transient state-prediction deviation relative to the steady-state demand and $D_T = \mathcal{O}(1)$ as the INAR(1) dynamics are stable.

It remains to bound the expected index error $\mathbb{E}[\epsilon_t]$. Under Assumption~\ref{assump: parameter} and the bounded martingale-difference noise condition, the projected ridge estimator achieves the following convergence rate:
\begin{equation}
    \begin{aligned}
        \left( \mathbb{E} \bigg[ \big\| \hat{\boldsymbol{\theta}}(t) - \boldsymbol{\theta}^* \big\|_{2,\infty}^2 \bigg] \right)^{1/2} = \mathcal{O}\left(\sqrt{\frac{\log(Bt)}{t}}\right) + \mathcal{O}\left(\frac{1}{t}\right).
    \end{aligned}
\end{equation}
Let $\mathcal{E}_t \triangleq \big\{ \| \hat{\boldsymbol{\theta}}(t) - \boldsymbol{\theta}^* \|_{2,\infty} \le r_{\boldsymbol{\theta}} \big\}$ denote the local concentration event. Conditioned on $\mathcal{E}_t$, Assumption~\ref{ass:locmw_index_regularity} guarantees that the induced LocMW index is locally stochastic Lipschitz, yielding
\begin{equation}
    \epsilon_t = \big\| \boldsymbol{\Phi}(t; \hat{\boldsymbol{\theta}}(t)) - \boldsymbol{\Phi}(t; \boldsymbol{\theta}^*) \big\|_\infty \le \Gamma_d(t) \big\| \hat{\boldsymbol{\theta}}(t) - \boldsymbol{\theta}^* \big\|_{2,\infty}.
\end{equation}
\rev{Moreover, the compactness of $\boldsymbol{\Theta}$, the positive denominator $1-\rho_b+\rho_b\eta_b^R(\theta)\ge 1-\rho_{\max}>0$, the finite window $d$, and the bounded AoI second moment collectively ensure that the index map has a uniformly bounded second moment $\sup_{t\ge 1}\mathbb{E}\left[\sup_{\theta\in\Theta^B}\|\Phi(t;\theta)\|_\infty^2\right]<\infty.$
Consequently, on the complement event $\mathcal{E}_t^c$, the error is constrained, rendering its expected contribution negligible. Taking the expectation and applying the Cauchy-Schwarz inequality, the expected index error is bounded by}
\begin{equation}
   \mathbb{E}[\epsilon_t] \le L_d^{\mathrm{L}} \cdot \mathcal{O}\left(\sqrt{\frac{\log(Bt)}{t}}\right) + \mathcal{O}\left(\frac{1}{t}\right). 
\end{equation}
Summing over $T$ yields
\begin{equation}
    \sum_{t=1}^{T} \mathbb{E}[\epsilon_t] = \mathcal{O}\left( L_d^{\mathrm{L}} \sqrt{T \log(BT)} \right) + \mathcal{O}(\log T).
\end{equation}
Finally, absorbing $D_T = \mathcal{O}(1)$ and $\mathcal{O}(\log T)$ into the $o(T)$, we obtain
\begin{equation}
    R_T^{\mathrm{LocMW}} = \mathcal{O}\left( K L_d^{\mathrm{L}} \sqrt{T\log(BT)} \right) + o(T),
\end{equation}
where $R_T^{\mathrm{LocMW}} \triangleq \big[ \sum_{t=1}^T \mathbb{E} \big[ \sum_{b \in \mathcal{B}} A_b(t) \big] - T C^{\mathrm{R},*} \big]^+$. This completes the proof.
\end{IEEEproof}



\section{Experiments}\label{sec:simulation}
In this section, we present numerical experiments to evaluate the proposed LocMW scheduling framework, comparing it against several baseline methods to demonstrate its superiority in AoI reduction and object detection under communication constraints.

\subsection{Experiment Settings}
\textbf{Setting.}
While the LocMW framework is broadly applicable to diverse wireless systems with local sensing and AoI-aware scheduling, we instantiate and evaluate it in the context of V2I-CP, the primary motivating application of this work.
Since standard CP datasets are specifically tailored for ego-user object detection, the number of users (vehicles) per scene is inherently limited (e.g., up to five in the V2X-Sim dataset). To assess the effectiveness of our framework in dense traffic, we utilize the pNEUMA \cite{barmpounakis2020new} and FLUID  \cite{Chen2026} datasets to simulate dense urban conditions.
The pNEUMA dataset provides large-scale, drone-captured vehicle trajectories over congested Athens. We use four days of recordings from drones d6 and d7, concatenating five consecutive 30-minute blocks per day and partitioning the coverage region into 543 areas via OpenStreetMap. The FLUID dataset offers fine-grained trajectories at signalized intersections. We treat each video sequence as a scene, retain only motorized users, and divide the region into $B = 201$ areas. Time is discretized into $0.1$s slots, mapping each vehicle to one area per slot. A vehicle's region of interest is defined as a radius of 80 m for pNEUMA and 60 m for FLUID.
Following distance-dependent sensing degradation and stochastic blockage models in vehicular environments~\cite{meireles2010experimental,hu2022where2comm}, We model the visibility of area $b$ to vehicle $u$ based on distance and local traffic density $v_{u,b}(t) = \exp\left[-\varepsilon\left(\frac{d_{u,b}(t)}{R_I} + \frac{n_b(t)}{n_{\max}}\right)\right]$ where $d_{u,b}(t)$ is the distance, $n_b(t)$ is the vehicle count in area $b$, and $n_{\max}$ is the maximum observed density.
We set $\zeta = 1$, $\rho_{\max} = 0.99$, and implement a 500-slot warm-up.

In addition, we evaluate the framework under realistic perception conditions using the V2X-Sim dataset \cite{li2022v2x}, a comprehensive V2X CP dataset simulated via SUMO and CARLA, and we randomly select 10 scenes. For the 3D object detection task, we adopt PointPillars \cite{lang2019pointpillars, ma2026birdcast} as the backbone detector. The feature maps are divided into $B = 400$ grids, with a data volume of $\delta = 1$ KB per grid. We comprehensively evaluate the system's performance based on both AoI and mAP.

\begin{figure*}[t]
\centering
\begin{subfigure}[b]{0.32\textwidth} 
  \includegraphics[width=\textwidth]{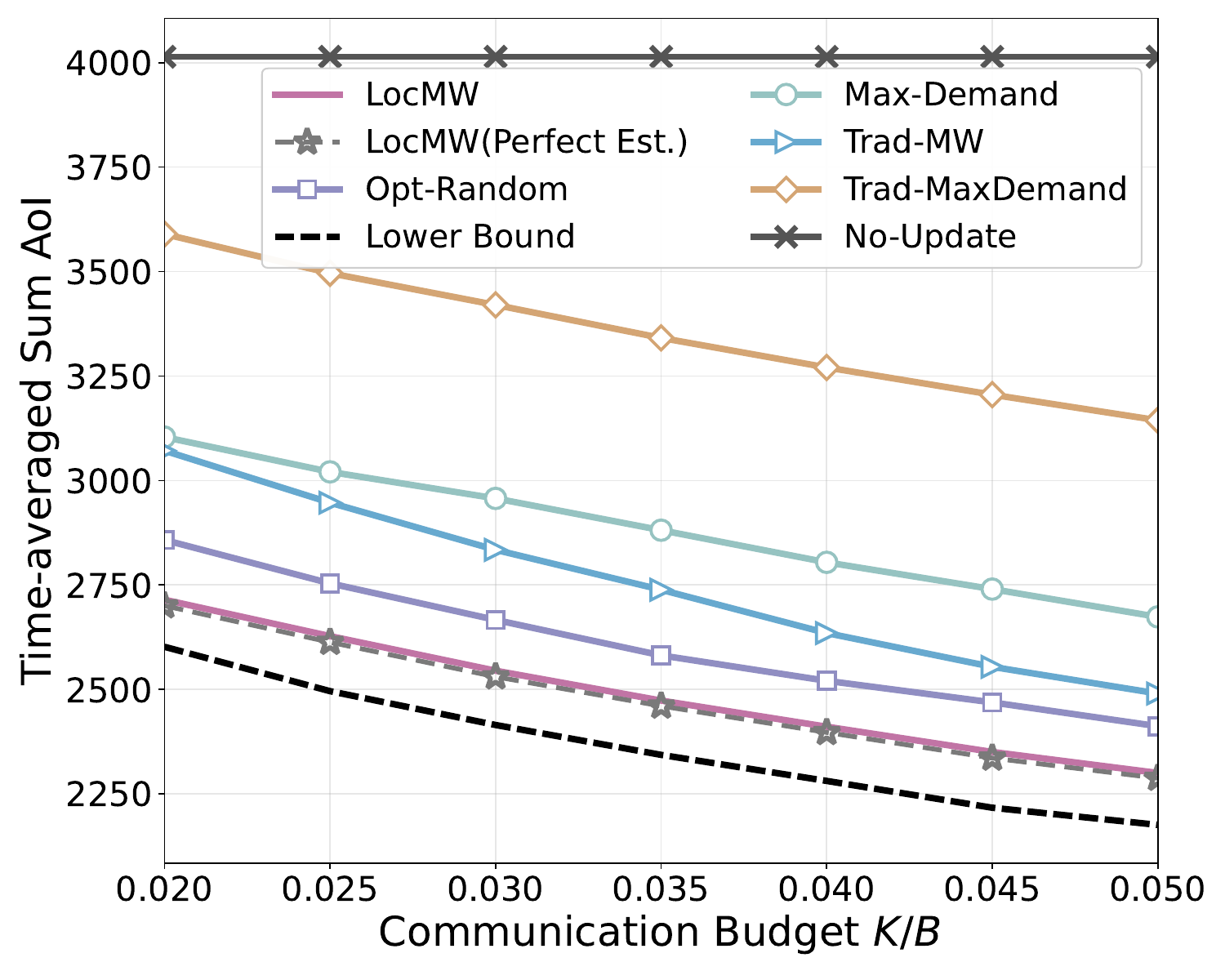}
  \caption{Time-average sum AoI versus communication budget $K/B$ ($d = 8$).} \label{fig:aoi_comm_1}
\end{subfigure}
\begin{subfigure}[b]{0.32\textwidth} 
  \includegraphics[width=\textwidth]{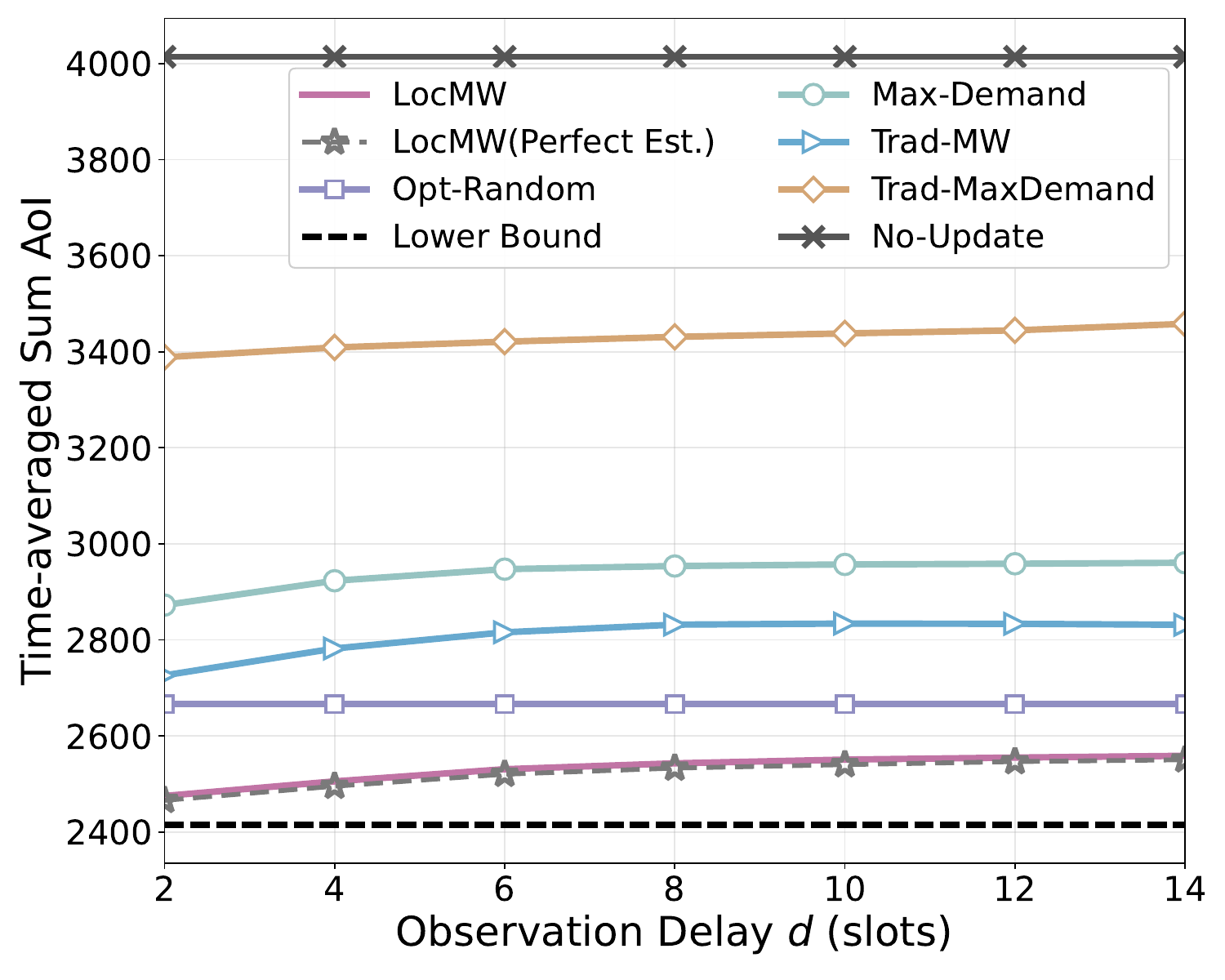}
  \caption{Time-average sum AoI versus observation delay $d$ in slot ($K/B = 0.03$).}\label{fig:aoi_delay_1}
\end{subfigure}
\begin{subfigure}[b]{0.32\textwidth} 
  \includegraphics[width=\textwidth]{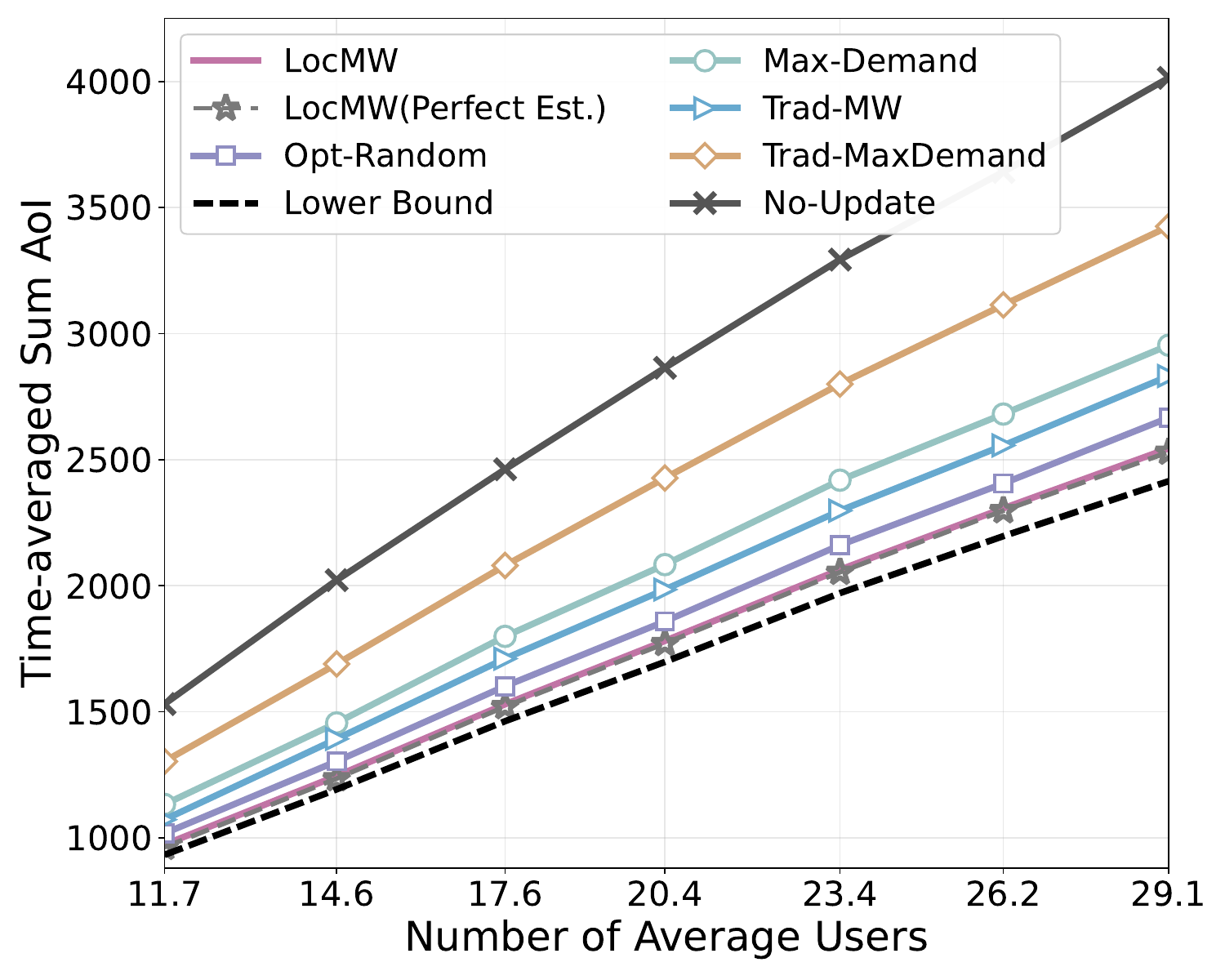}
  \caption{Time-average sum AoI versus number of users ($d = 8$, $K/B = 0.03$).} \label{fig:aoi_num_users_1}
\end{subfigure}
\caption{Performance evaluation of the time-average sum AoI under varying system settings using the pNEUMA dataset.}
\end{figure*}

\begin{figure*}[t]
\centering
\begin{subfigure}[b]{0.32\textwidth} 
  \includegraphics[width=\textwidth]{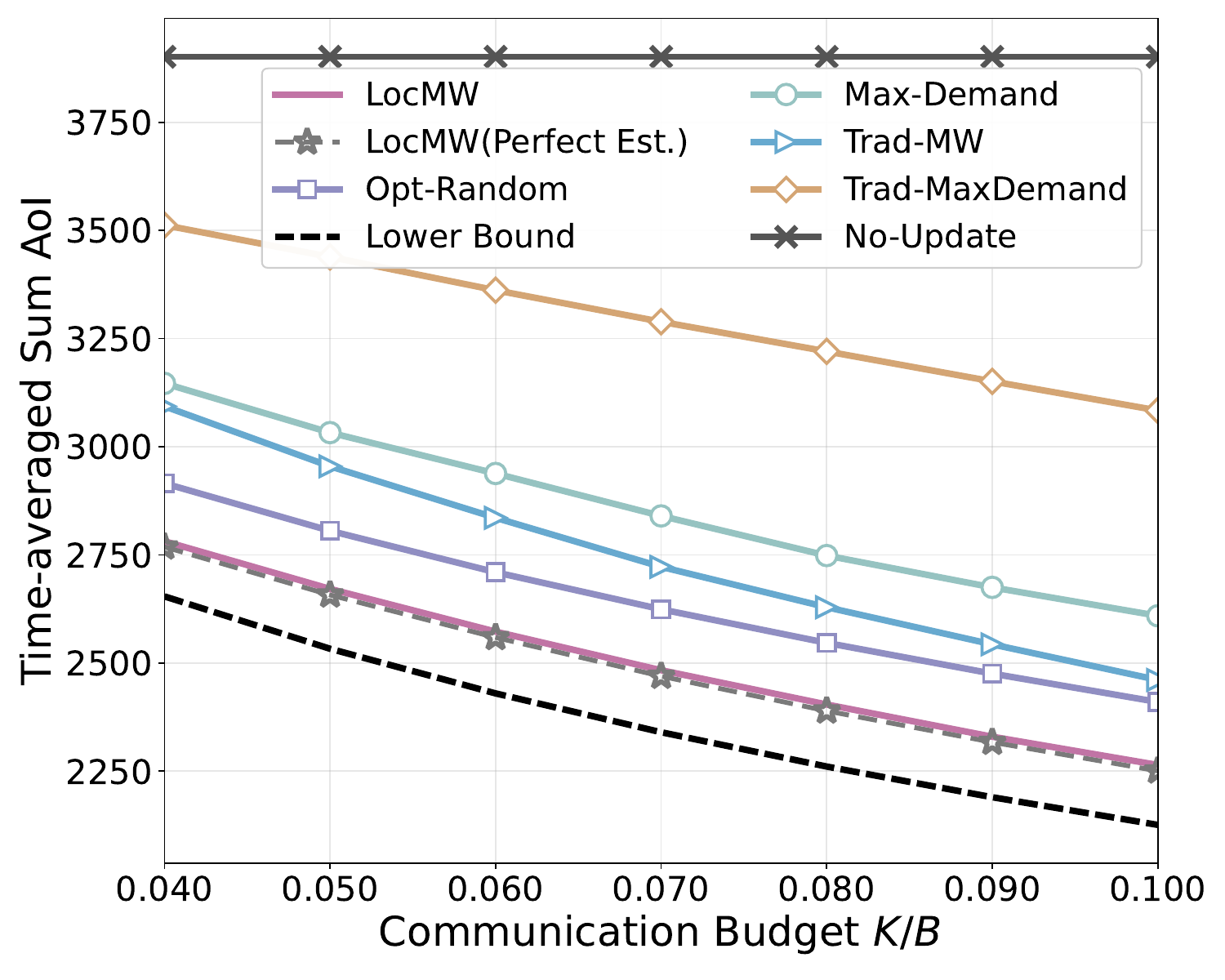}
  \caption{Time-average sum AoI versus communication budget $K/B$ ($d = 8$).} \label{fig:aoi_comm_2}
\end{subfigure}
\begin{subfigure}[b]{0.32\textwidth} 
  \includegraphics[width=\textwidth]{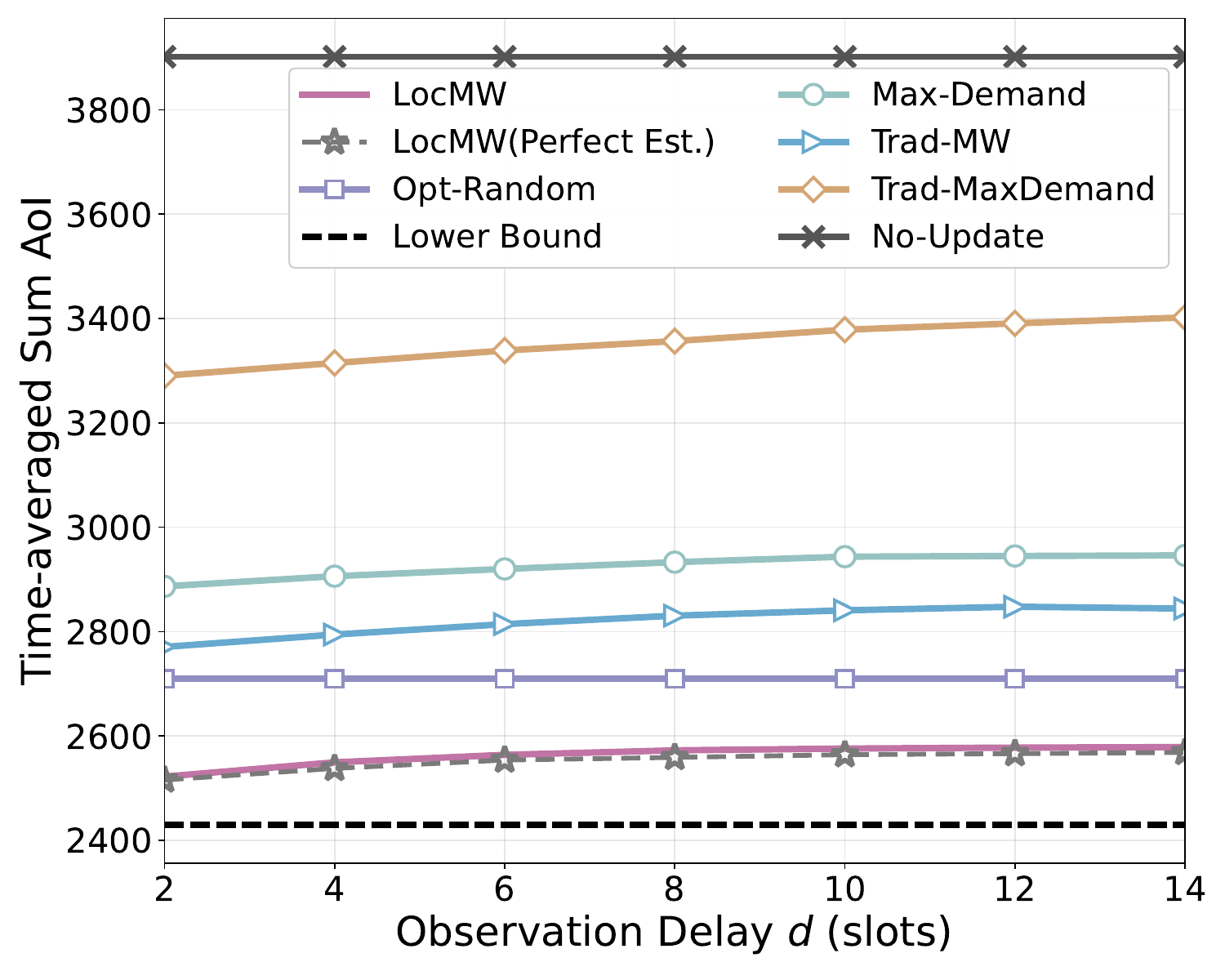}
  \caption{Time-average sum AoI versus observation delay $d$ in slot ($K/B = 0.06$).}\label{fig:aoi_delay_2}
\end{subfigure}
\begin{subfigure}[b]{0.32\textwidth} 
  \includegraphics[width=\textwidth]{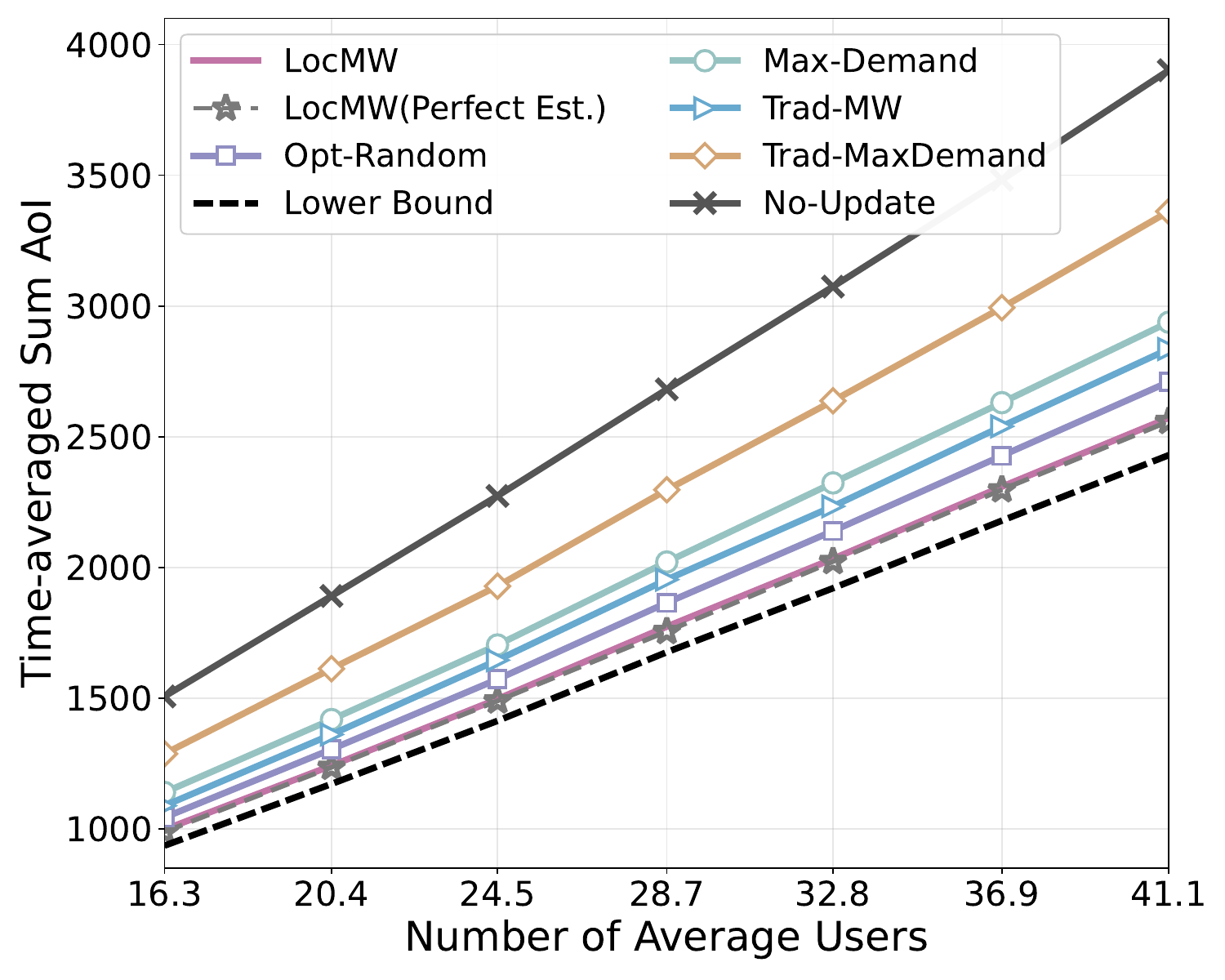}
  \caption{Time-average sum AoI versus number of users ($d = 8$, $K/B = 0.06$).} \label{fig:aoi_num_users_2}
\end{subfigure}
\caption{Performance evaluation of the time-average sum AoI under varying system settings using the FLUID dataset.} \label{fig:fluid_results}
\end{figure*}

\textbf{Baselines.}
We compare the proposed LocMW policy against several baselines. 

\begin{itemize}
\item \textbf{LocMW (Perfect Est.):}
This oracle benchmark assumes the BS has perfect knowledge of the true parameters ($\rho_b,\mu_b$), isolating the loss caused by online estimation.

\item \textbf{Max-Demand:}
This policy schedules \(K\) areas with the highest predicted demand \(\hat N_b(t)\). It prioritizes areas with the most vehicles requiring BS updates without AoI state.

\item \textbf{Traditional Max-Weight:}
This policy schedules the $K$ areas that maximize the index $\hat\rho_b\tilde A_b^{\mathrm T}(t)/(1-\hat\rho_b+\hat\rho_b\hat\eta_b^{\mathrm R})$, where $\tilde A_b^{\mathrm T}$ is the AoI state without accounting for local sensing. 

\item \textbf{Traditional Max-Demand:}
This policy schedules $K$ areas with the highest traditional demand estimates, $\tilde N_b^{\mathrm T}(t)$, without accounting for local data sensing, reflecting the behavior of a conventional demand-based scheduler.

\item \textbf{Stationary Randomized Policy:}
The stationary randomized benchmark developed in Section~\ref{sec:randomized_policy}. 

\item \textbf{Mean-field Lower Bound:}
The analytical lower bound, $C^{\mathrm{LB}}$, derived in Section~\ref{sec:lower_bound}.

\item \textbf{No Update:}
This policy schedules no BS transmissions, where AoI depends on users' local sensing only.

\end{itemize}

\begin{figure*}[t]
    \begin{minipage}[b]{0.48\textwidth}
        \centering
        \begin{subfigure}[b]{0.49\linewidth} 
            \includegraphics[width=\linewidth]{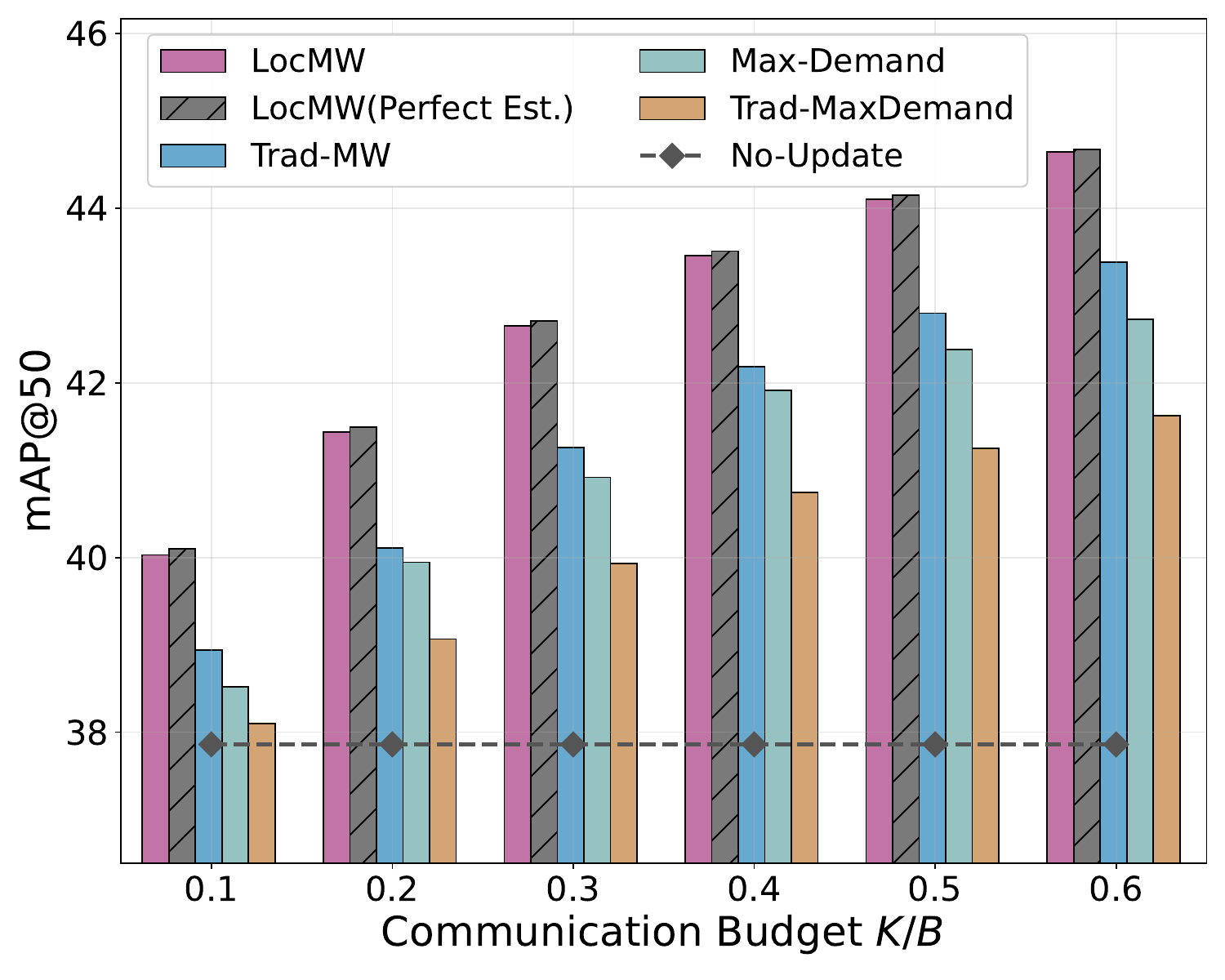}
            \caption{mAP@50 versus the communication budget $K/B$.}\label{fig:map50_vs_bandwidth}
        \end{subfigure}
        \hfill
        \begin{subfigure}[b]{0.49\linewidth} 
            \includegraphics[width=\linewidth]{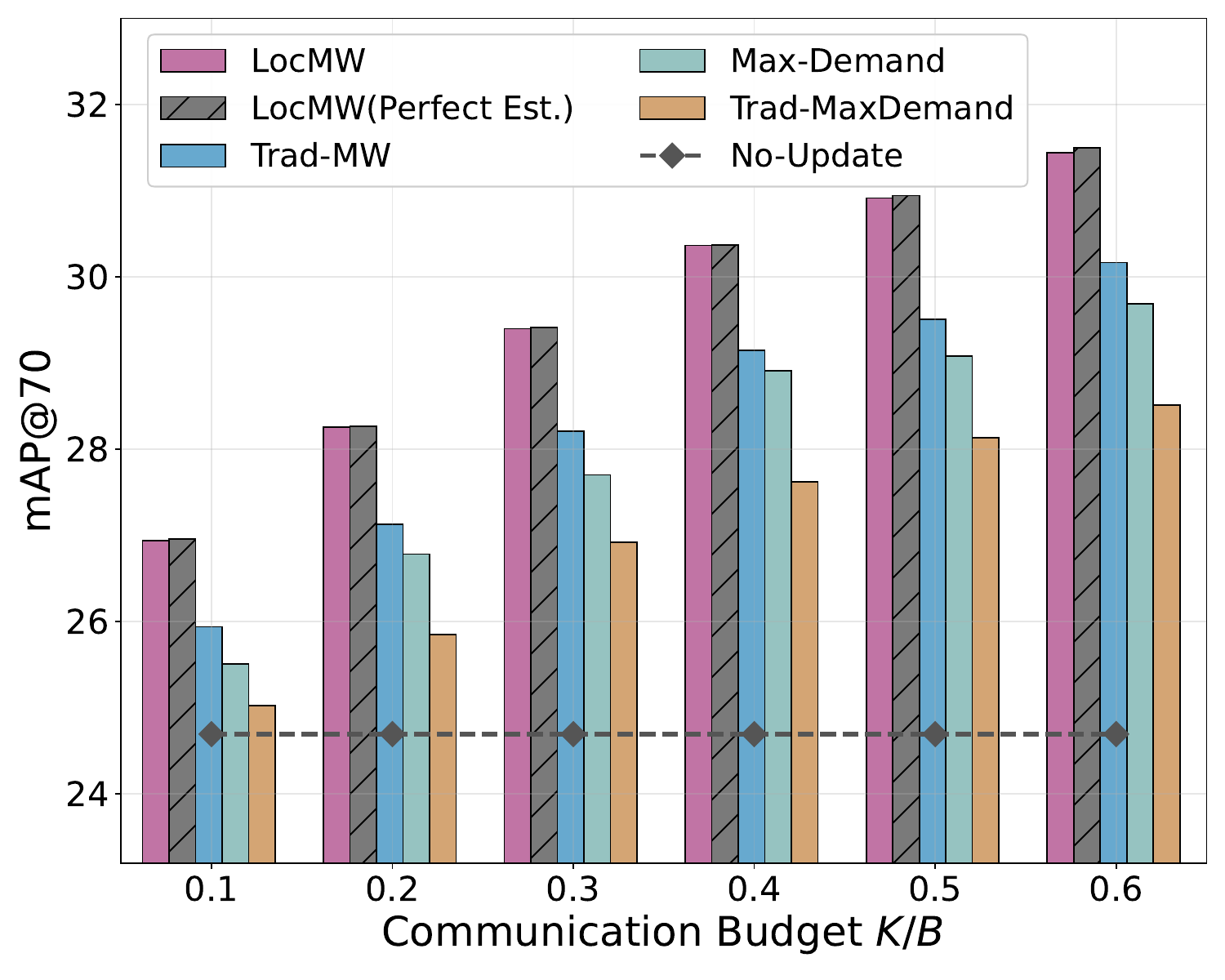}
            \caption{mAP@70 versus the communication budget $K/B$.}\label{fig:map70_vs_bandwidth}
        \end{subfigure}
        \caption{The perception performance (mAP@50/mAP@70) versus the communication budget $K/B$ with observation delay $d=8$ slots.}\label{fig:map_vs_bandwidth}
    \end{minipage}
    \hfill
    \begin{minipage}[b]{0.48\textwidth}
        \centering
        \begin{subfigure}[b]{0.49\linewidth} 
            \includegraphics[width=\linewidth]{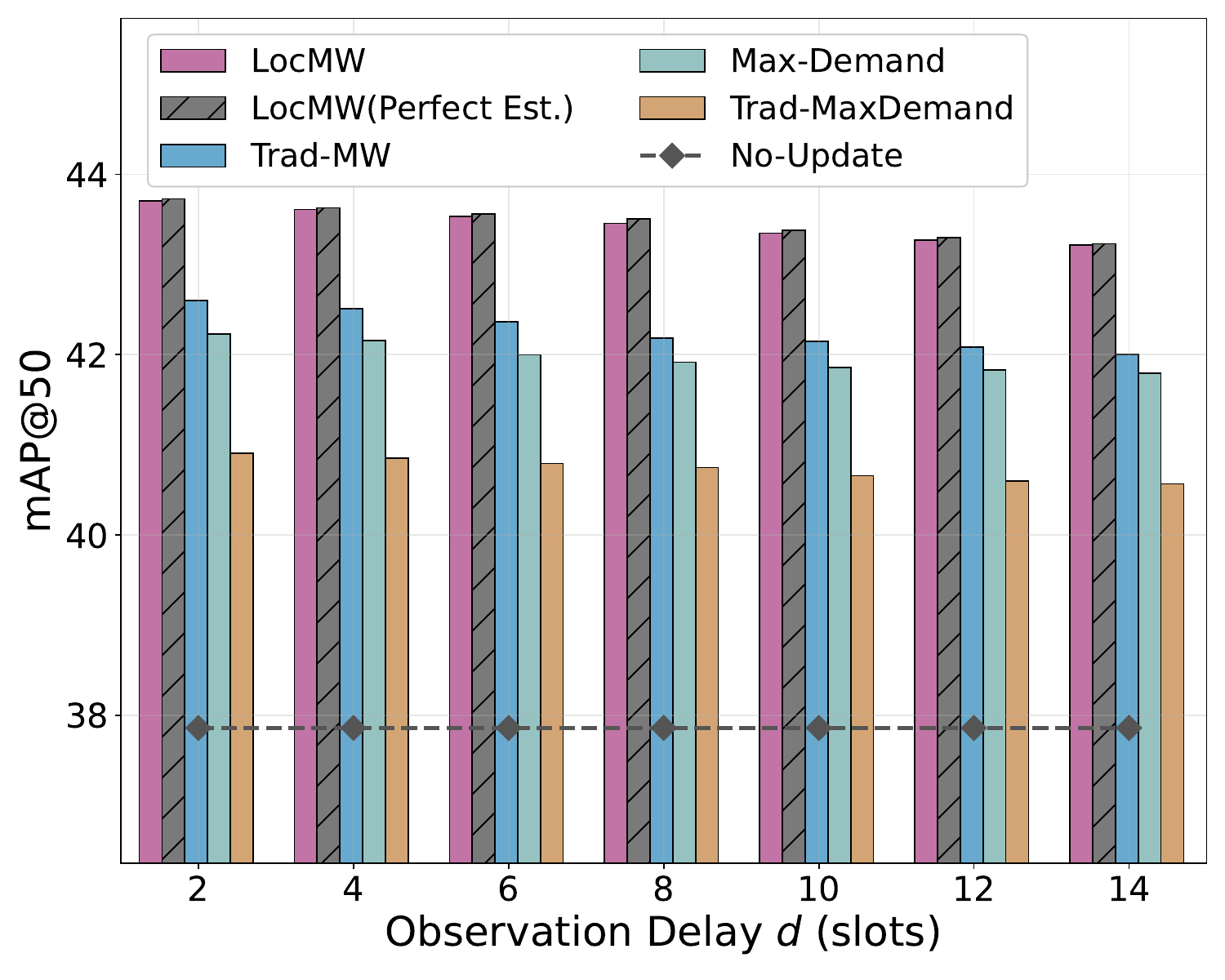}
            \caption{mAP@50 versus the observation delay $d$.}\label{fig:map50_vs_latency}
        \end{subfigure}
        \hfill
        \begin{subfigure}[b]{0.49\linewidth}
            \includegraphics[width=\linewidth]{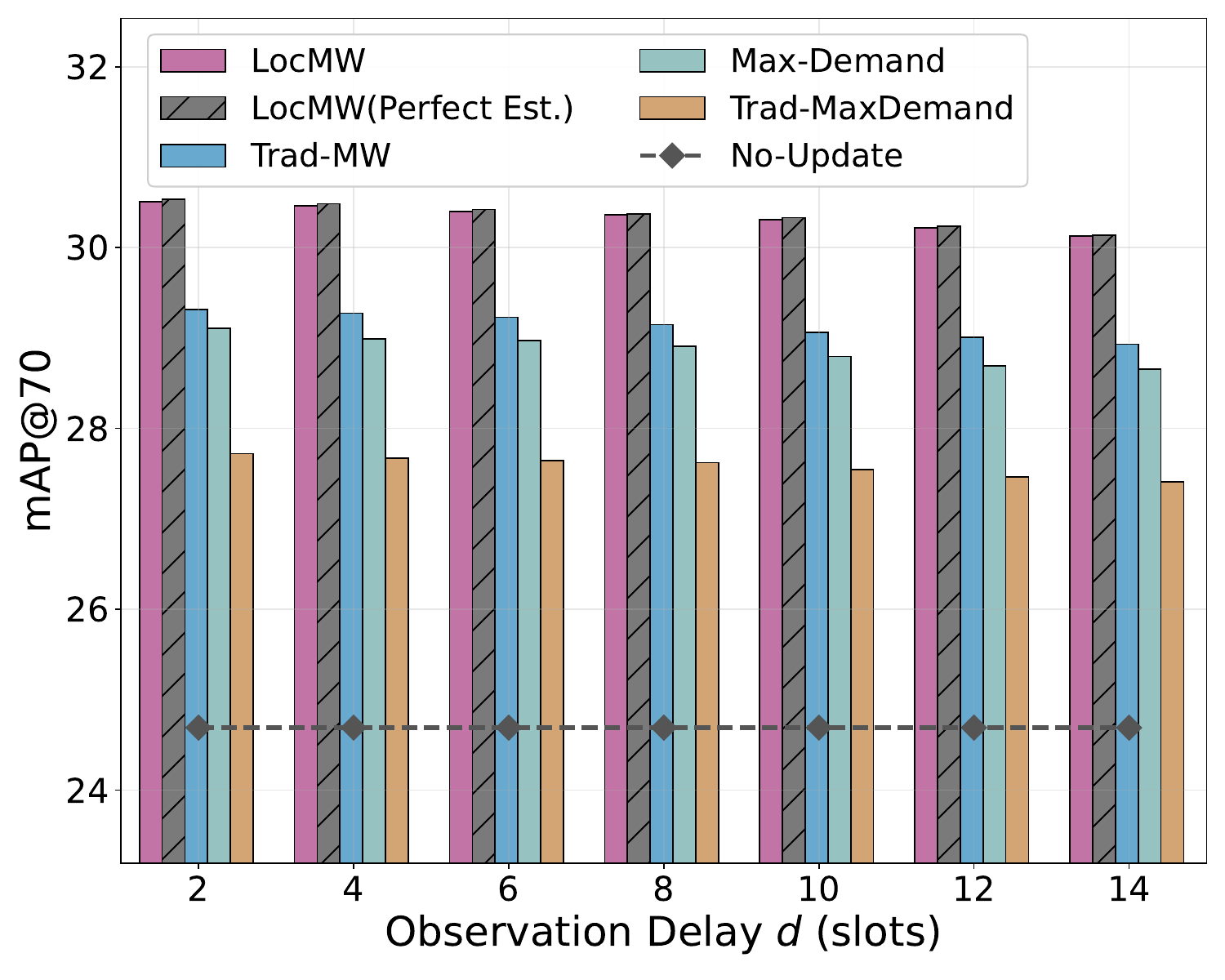}
            \caption{mAP@70 versus the observation delay $d$.} \label{fig:map70_vs_latency}
        \end{subfigure}
        \caption{The perception performance (mAP@50/mAP@70) versus the observation delay $d$ with communication budget $K/B$.}\label{fig:map_vs_latency}
    \end{minipage}

    \vspace{0.4cm} 

    \centering
    \includegraphics[width=0.96\textwidth]{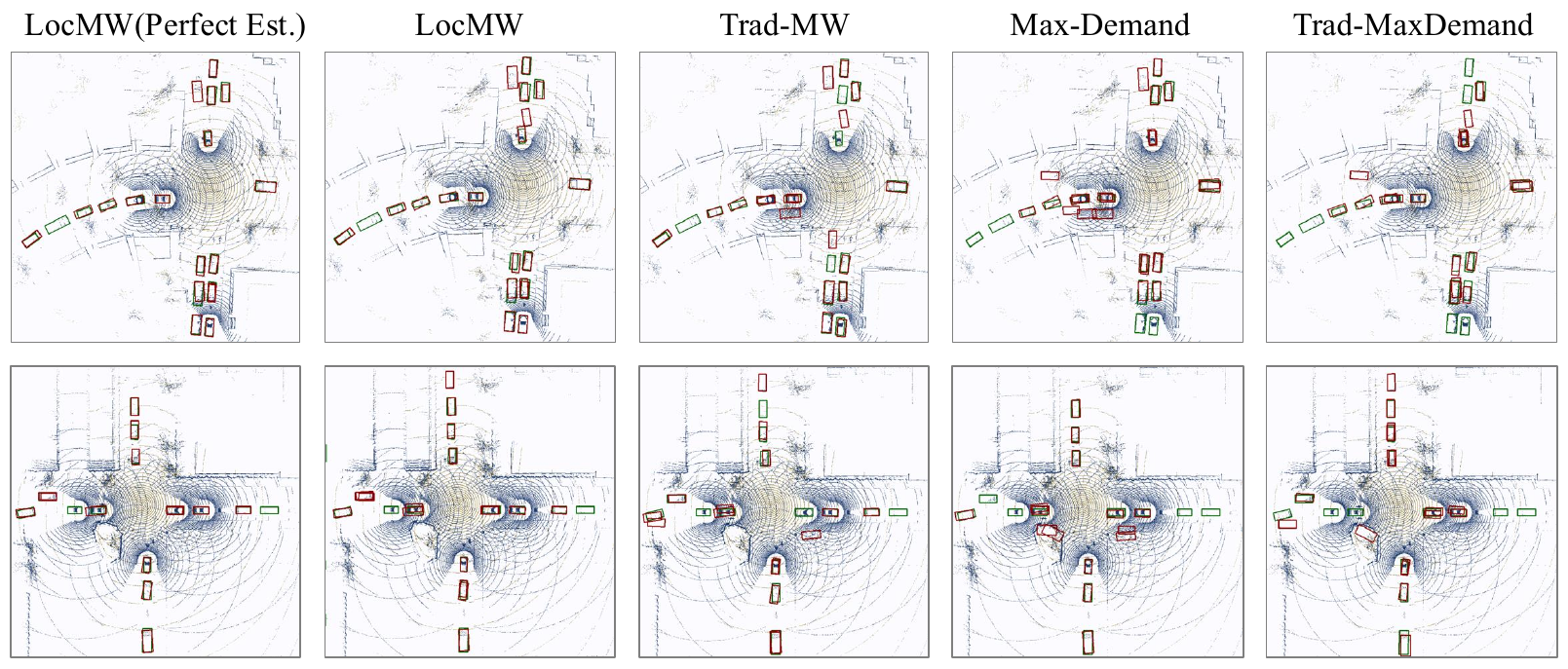}
    \caption{Qualitative visualization of 3D object detection on the V2X-Sim dataset, comparing the proposed LocMW policy against baseline scheduling schemes. \textcolor{darkred}{Red} bounding boxes denote predictions, while \textcolor{darkgreen}{green} ones represent the ground truth (GT).}\label{fig:vis}
    \vspace{-0.0cm}
\end{figure*}

\subsection{AoI Evaluation}
Fig.~\ref{fig:aoi_comm_1} illustrates the time-average sum AoI under various communication budgets $K/B$. As expected, increasing the communication budget consistently reduces the AoI for all update-based policies by allowing more areas to be refreshed per slot. \rev{The proposed LocMW policy achieves the lowest AoI among all online policies, performing comparably to the benchmark with perfect estimation. Notably, it reduces the sum AoI by up to 31.6\% compared to the Traditional Max-Demand baseline.}
This demonstrates that the online parameter estimation incurs negligible performance degradation. Although the stationary randomized policy outperforms the heuristic baselines, it remains inferior to LocMW because it cannot adapt to instantaneous AoI states or dynamic local sensing events. In contrast, the traditional max-weight policy is less effective because it neglects local-sensing-aware AoI dynamics, wasting limited bandwidth on redundant BS updates for areas already refreshed by onboard vehicular sensors. Furthermore, while Max-Demand effectively prioritizes areas with high user concentrations, its disregard for accumulated freshness degradation leads to suboptimal performance. Finally, the No Update scheme yields the highest AoI, underscoring the importance of BS updates in maintaining system-wide perception freshness.

Fig.~\ref{fig:aoi_delay_1} evaluates the impact of the observation delay $d$. The AoI of LocMW exhibits only a slight increase as $d$ grows, demonstrating its robustness to delayed state observations. This is because LocMW predicts unobserved system states and schedules areas based on the expected reduction in AoI, rather than relying only on stale observations. Moreover, the performance gap between LocMW and its perfect-estimation counterpart remains marginal across all values of $d$. The stationary randomized policy is insensitive to $d$ as it relies on long-term statistics. Traditional Max-Weight and Traditional Max-Demand yield significantly higher AoI, indicating that ignoring the local sensing process or the accumulated AoI leads to suboptimal decisions under delayed observations.

Fig.~\ref{fig:aoi_num_users_1} shows the performance of the proposed policy versus the average number of users. As expected, the time-average sum AoI increases for all policies as the user population grows. LocMW consistently achieves the best performance among the online policies and maintains a marginal gap to the theoretical lower bound across the whole range of user densities. Moreover, its performance advantage over Max-Demand and Traditional Max-Weight becomes more pronounced as the number of users increases, because the baselines' inaccurate modeling of local sensing leads to increasingly inefficient utilization of the limited communication budget. 

The simulation results on the FLUID dataset are shown in Fig.~\ref{fig:fluid_results}, which exhibit consistent trends with those observed in the pNEUMA dataset. 
Across all settings, LocMW consistently attains the lowest AoI among the online policies and performs comparably to its perfect-estimation counterpart. These results demonstrate that the proposed LocMW framework generalizes effectively to diverse real-world traffic traces, delivering robust AoI reduction under varying bandwidth budgets, observation delays, and traffic densities.


\subsection{Perception Evaluation}
Fig.~\ref{fig:map_vs_bandwidth} presents the 3D object detection accuracy on the V2X-Sim dataset under varying communication budget $K/B$. As the communication budget increases, both mAP@50 and mAP@70 improve across all update-based schemes, as a greater volume of features can be refreshed in each slot. 
\rev{The proposed LocMW policy consistently achieves the highest mAP-improving mAP@70 by up to 16.3\% relative to the Traditional Max-Demand baseline, and performs comparably to the perfect-estimation benchmark.}

In contrast, traditional Max-Weight and traditional Max-Demand exhibit inferior detection accuracy because they neglect local sensing-induced AoI reductions and tend to allocate bandwidth to grids that have already been refreshed locally. In addition, Max-Demand performs suboptimally as it prioritizes user density without explicitly accounting for accumulated information staleness. The No-Update baseline yields the poorest performance, confirming that BS updates are essential for collaborative perception.

Fig.~\ref{fig:map_vs_latency} shows the impact of observation delay on perception accuracy. As the delay increases, the mAP of most scheduling policies degrades, given that the BS makes decisions based on increasingly outdated network-state information. Nevertheless, LocMW demonstrates remarkable stability across different delay periods and consistently outperforms the heuristic and traditional baselines. These results show that the proposed local-sensing-aware AoI scheduling framework benefits task-level perception quality.

Fig.~\ref{fig:vis} provides qualitative visualization results of object detection on the V2X-Sim dataset. The proposed LocMW produces detection results that are more consistent with ground truth, showing better spatial alignment and fewer missed vehicles. In contrast, Trad-MW, Max-Demand, and Traditional-MaxDemand exhibit more incomplete or less accurate detections, particularly for vehicles located in peripheral or occluded regions. These visualizations confirm that local-sensing-aware AoI scheduling improves not only information freshness but also downstream collaborative perception quality.



\section{Conclusion}
\label{sec:conclusion}
In this paper, we have investigated AoI minimization for infrastructure-assisted collaborative perception under limited downlink bandwidth. Unlike conventional AoI scheduling optimizations, the proposed framework explicitly accounts for vehicles' local sensing capabilities, under which information freshness can be improved not only through BS broadcasts but also through users' own perception. To capture this distinctive feature, we have modeled the interested-but-unobserving user population as an INAR(1) process and derived a closed-form characterization of the time-average aggregate AoI. Based on this characterization, we have established a mean-field lower bound and an optimal stationary randomized benchmark. 
We then proposed LocMW, a local-sensing-aware Max-Weight scheduling policy. We have also provided the theoretical analysis, showing that LocMW incurs only sublinear cumulative excess AoI relative to the optimal randomized policy. Extensive experiments on pNEUMA, FLUID, and V2X-Sim datasets have further demonstrated that LocMW consistently reduces information staleness and improves downstream 3D object detection accuracy compared with competing baselines.



\appendices
\section{Proof of Lemma \ref{lemma:one_step_drift}}
\label{proof:one_step_drift}
Conditioned on the filtration $\mathcal{F}_t$, the scheduling decision $u_b(t)$, the demanding user count $N_b(t)$, and the individual AoI $a_{b,i}(t)$ are deterministic. Conversely, the indicators $X_i(t)$ and the new arrivals $\omega_b(t+1)$ are independent of $\mathcal{F}_t$. Thus, taking the conditional expectation $\mathbb{E}[\cdot \mid \mathcal{F}_t]$ of the AoI evolution in \eqref{eq:aoi_def} yields:
\begin{equation}
    \begin{aligned}
        \mathbb{E}[A_b(t+1) | \mathcal{F}_t] =& \mathbb{E} \Big[ \sum_{i=1}^{N_b(t)} X_i(t) ((1 - u_b(t)) a_{b,i}(t) + 1) | \mathcal{F}_t \Big]\\
        &+ \mathbb{E} \Big[\omega_b(t+1) | \mathcal{F}_t \Big]\\
        =& \sum_{i=1}^{N_b(t)} \rho_b \big[ (1 - u_b(t)) a_{b,i}(t) + 1 \big] + \mu_b\\
        =& \rho_b \big(1 - u_b(t)\big) A_b(t) + \rho_b N_b(t) + \mu_b,
    \end{aligned}
\end{equation}
which completes the proof.

\section{Proof of Lemma \ref{lemma:stability}}
\label{proof:stability}
We first analyze the worst-case scenario where the BS never schedules an update for area $b$, i.e., $u_b(t) \equiv 0$ for all $t$. Given the physical bounds $N_b(t) \le N_{\max}$ and $\mu_b \le \mu_{\max}$, \rev{applying Lemma \ref{lemma:one_step_drift} and the law of iterated expectations yields}
\begin{equation}
    \mathbb{E}[A_b(t+1)] \le \rho_b \mathbb{E}[A_b(t)] + \rho_b N_{\max} + \mu_{\max}.
\end{equation}
Defining $C \triangleq \rho_b N_{\max} + \mu_{\max}$ and unrolling this recursion from $t=1$, we obtain
\begin{equation}
\begin{aligned}
    \mathbb{E}[A_b(t)] &\le \rho_b^{t-1} \mathbb{E}[A_b(1)] + C \sum_{k=0}^{t-2} \rho_b^k\\
    &\overset{\text{(a)}}{\le} \mathbb{E}[A_b(1)] + \frac{\rho_b N_{\max} + \mu_{\max}}{1-\rho_b} \triangleq M < \infty.
\end{aligned}
\end{equation}
where inequality $(a)$ holds because $\rho_b \in [0, 1)$.
Since any admissible scheduling policy $u_b(t) \in \{0,1\}$ either maintains or reduces AoI compared to this passive baseline, the upper bound $\sup_{t \ge 1} \mathbb{E}[A_b(t)] \le M$ holds universally. Consequently, the time-averaged boundary difference satisfies
\begin{equation}
    0 \le \lim_{T \to \infty} \frac{1}{T} \big| \mathbb{E}[A_b(T+1)] - \mathbb{E}[A_b(1)] \big| \le \lim_{T \to \infty} \frac{2M}{T} = 0, 
\end{equation}
which directly implies
\begin{equation}
    \lim_{T\to\infty}\frac{1}{T}\big(\mathbb{E}[A_{b}(T+1)]-\mathbb{E}[A_{b}(1)]\big)=0,
\end{equation}
completing the proof.

\section{Proof of Theorem \ref{thm:lower_bound}}
\label{proof:lower_bound}
Conditioning the aggregate AoI recursion in \eqref{eq:aoi_def} on \(\mathcal{H}(t)\), and noting that \(u_b(t)\) is \(\mathcal{H}(t)\)-measurable, gives
\begin{equation} \label{eq:z_evolution}
\begin{aligned}
    Z_b^+(t+1)&=\mathbb{E}[A_b(t+1)|\mathcal{H}(t)] \\
    &= \rho_b(1-u_b(t))Z_b(t)+\tilde{M}_b(t+1). 
\end{aligned}
\end{equation}
Taking the time average of both sides yields
\begin{equation}
    z_b=\rho_b(z_b-x_b)+\lambda_b,
    \qquad
    x_b=\frac{\lambda_b-(1-\rho_b)z_b}{\rho_b}.
\end{equation}
\rev{Squaring both sides of \eqref{eq:z_evolution} and considering $(1-u_b(t))^2 = 1-u_b(t)$ as $u_b(t) \in \{0,1\}$, we obtain
\begin{equation} \label{eq:sqare_z}
    \begin{aligned}
        Z_b^2(t+1) = &\rho_b^2(1-u_b(t))Z_b^2(t)+ \tilde{M}_b^2(t+1) \\
        &+ 2\rho_b(1-u_b(t))Z_b(t)\tilde{M}_b(t+1).    
    \end{aligned}
\end{equation}
Let $Q_b(t) \triangleq (1-u_b(t))Z_b(t)$, we have
\begin{equation} \label{eq:res}
    \begin{aligned}
        &\frac{1}{T}\sum_{t=1}^T \mathbb{E}[Q_b(t)\tilde{M}_b(t+1)] - \lambda_b\frac{1}{T}\sum_{t=1}^T \mathbb{E}[Q_b(t)] \\
        =& \frac{1}{T}\sum_{t=1}^T \mathbb{E}[Q_b(t)(\tilde{M}_b(t+1) - \lambda_b)].
    \end{aligned}
\end{equation}
Applying the Cauchy-Schwarz inequality to the right-hand term gives
\begin{equation}
    \begin{aligned}
        &\left| \frac{1}{T}\sum_{t=1}^T \mathbb{E}\left[Q_b(t)\big(\tilde{M}_b(t+1) - \lambda_b\big)\right] \right| \\
        \leq & \left( \frac{1}{T}\sum_{t=1}^T \mathbb{E}[Q_b^2(t)] \right)^{1/2} \left( \frac{1}{T}\sum_{t=1}^T \mathbb{E}[\big(\tilde{M}_b(t+1) - \lambda_b\big)^2] \right)^{1/2}.
    \end{aligned}\label{cauchy}
\end{equation}
Under Assumption~\ref{ass:mf_decoupling}, $\lim_{T\to\infty}\frac{1}{T}\sum_{t=1}^{T}
\mathbb{E}[(\tilde M_b(t+1)-\lambda_b)^2]=0$.
By Jensen's inequality, we have $\mathbb{E}[Q_b^2(t)] \le \mathbb{E}[Z_b^2(t)] \le \mathbb{E}[A_b^2(t)]$. Considering the passive worst-case policy $u_b(t) \equiv 0$ and the physical bounds $N_b(t) \le N_{\max}$ and $\omega_b(t) \le \mu_{\max}$, we can bound the conditional second moment of the one-step AoI evolution using Young's inequality as $\mathbb{E}[A_b^2(t+1) |\mathcal{F}_t] \le (1+\delta)\rho_b(A_b(t)+N_{\max})^2 + (1+1/\delta)\mu_{\max}^2$ for any $\delta>0$. Since $\rho_b < 1$, choosing a sufficiently small $\delta>0$ yields $\mathbb{E}[A_b^2(t+1)] \le \alpha_b \mathbb{E}[A_b^2(t)] + C_b$ for some constants $\alpha_b < 1$ and $C_b < \infty$. Unrolling this recursion guarantees $\sup_{t \ge 1} \mathbb{E}[A_b^2(t)] < \infty$, which implies that $\mathbb{E}[Q_b^2(t)]$ remains bounded. Thus, from \eqref{eq:res} and setting the left hand side of \eqref{cauchy} to zero, we have
\begin{equation}
    \begin{aligned}
        &\lim_{T\to\infty}\frac{1}{T}\sum_{t=1}^T \mathbb{E}\left[(1-u_b(t))Z_b(t)\tilde{M}_b(t+1)\right] \\
        =& \lambda_b \lim_{T\to\infty}\frac{1}{T}\sum_{t=1}^T \mathbb{E}\left[(1-u_b(t))Z_b(t)\right] = \lambda_b(z_b-x_b).
    \end{aligned}
\end{equation}

Define $y_b \triangleq \lim_{T\to\infty}\frac{1}{T}\sum_{t=1}^T \mathbb{E}[Z_b^2(t)]$ and $w_b \triangleq \lim_{T\to\infty}\frac{1}{T}\sum_{t=1}^T\mathbb{E}[u_b(t)Z_b^2(t)]$. By taking the time-average expectation of \eqref{eq:sqare_z} and invoking Assumption~\ref{ass:mf_decoupling}, we obtain
\begin{equation}
y_b = \rho_b^2(y_b - w_b) + 2\rho_b\lambda_b(z_b-x_b) + \lambda_b^2.
\end{equation}
Rearranging the terms and using $\rho_b(z_b-x_b) = z_b - \lambda_b$ from the first-moment equality, this simplifies to
\begin{equation}
(1-\rho_b^2)y_b+\rho_b^2 w_b = 2\rho_b\lambda_b(z_b-x_b)+\lambda_b^2 = 2\lambda_b z_b-\lambda_b^2.
\end{equation}
}
Moreover, Jensen's inequality gives \(y_b\ge z_b^2\) and Cauchy--Schwarz inequality gives \(w_b\ge x_b^2/p_b\). Substituting the expression for \(x_b\) results in
\begin{equation}
    (1-\rho_b^2)z_b^2+\frac{\big(\lambda_b-(1-\rho_b)z_b\big)^2}{p_b}
    \le2\lambda_b z_b-\lambda_b^2 .
\end{equation}
Solving it for $z_b$ establishes the lower bound:
\begin{equation}
    z_b\ge L_b(p_b) = \lambda_b\frac{p_b+1}{p_b(1+\rho_b)+1-\rho_b}.
\end{equation}

Since \(\mathbb{E}[Z_b(t)]=\mathbb{E}[A_b(t)]\), we have \(z_b=\bar A_b\) in the time-average sense. Therefore, every admissible policy satisfies
\begin{equation}
    \sum_{b=1}^B \bar A_b \ge \sum_{b=1}^B L_b(p_b).
\end{equation}
Minimizing the right-hand side over all feasible  \(p_b\) gives \(C^{\mathrm{LB}}\).
The Karush-Kuhn-Tucker (KKT) stationarity condition is
\begin{equation}
    \gamma=\frac{2\lambda_b\rho_b}{\big(p_b(1+\rho_b)+1-\rho_b\big)^2}.
\end{equation}
Solving for \(p_b\) and projecting onto \([0,1]\) yields \eqref{eq:lb_water_filling}. Since
\(\sum_b p_b^*(\gamma)\) is monotone in \(\gamma\), the multiplier can be efficiently found by bisection. This completes the proof.



\bibliographystyle{IEEEtran}
\bibliography{reference}


 




\vfill

\end{document}